\documentclass[fleqn,11pt,letterpaper]{article}

\usepackage{bbm}
\usepackage{appendix}
\usepackage{amsfonts}
\usepackage{amssymb}
\usepackage{amsmath}
\usepackage{amsthm}
\usepackage{enumitem}
\usepackage{multirow}
\usepackage{float}
\usepackage{placeins}
\usepackage{graphicx}
\usepackage{epstopdf}
\usepackage{color}
\usepackage[margin=2.57cm]{geometry}
\definecolor{darkred}{rgb}{0.5,0.2,0.2}
\usepackage[pdftitle={Monitoring the Global Carbon Budget},pdfauthor={Mikkel Bennedsen},colorlinks=TRUE,allcolors=darkred]{hyperref}
\usepackage{booktabs}
\usepackage{pdflscape}
\usepackage[round]{natbib}
\usepackage{subcaption}
\usepackage{comment}
\usepackage{todonotes}

\usepackage{url}
\usepackage{systeme}

\theoremstyle{mytheorem}

\theoremstyle{myremark}

\nonstopmode

\usepackage{tabularx}

\usepackage{float}
\floatstyle{plaintop}
\restylefloat{table}

\restylefloat{table}

\theoremstyle{plain}
\newtheorem{theorem}{Theorem}[section]

\newtheorem{assumption}{Assumption}[section]

\newtheorem{lemma}{Lemma}[section]

\theoremstyle{definition}

\theoremstyle{remark}
\newtheorem{remark}{Remark}[section]

\def\bi{\begin{itemize}}
	\def\ei{\end{itemize}}

\allowdisplaybreaks
\graphicspath{{Figures/}}
\numberwithin{equation}{section}

\newif\ifi

\usepackage{lineno}

\begin{document}
	
\title{Propensity score with factor loadings: evaluating the effect of the Paris
	Agreement on the returns of European stocks
} %

\author{
Angelo Forino\thanks{
Department of Statistical Sciences, 
Sapienza University of Rome, 
00185 Rome, Italy. 
E-mail:
\url{angelo.forino@uniroma1.it}.}
\and
Andrea Mercatanti\thanks{
Department of Statistical Sciences, 
Sapienza University of Rome, 
00185 Rome, Italy. 
E-mail: \url{andrea.mercatanti@uniroma1.it}.}
\and 
Giacomo Morelli\thanks{
Department of Statistical Sciences, 
Sapienza University of Rome, 
00185 Rome, Italy.
E-mail: \url{giacomo.morelli@uniroma1.it}}
}

\maketitle

\begin{abstract}
Factor models for longitudinal data, where policy adoption is unconfounded with respect to a low-dimensional set of latent factor loadings, have become increasingly popular for causal inference. Most existing approaches, however, rely on a causal finite-sample approach or computationally intensive methods, limiting their applicability and external validity. In this paper, we propose a novel method for panel causal inference within a super-population framework, based on inverse propensity score weighting, where the propensity score is a function of latent factor loadings. The approach relaxes the traditional restrictive assumptions of causal panel methods, while offering advantages in terms of causal interpretability, policy relevance, and computational efficiency. Under standard assumptions, 
we develop a three-step estimation procedure for the ATT and derive its large-sample properties using M-estimation theory. Simulation results show that the proposed estimator remains reliable across a variety of designs, including settings with misspecified factor structure and weakly dependent errors. We apply the method to assess the causal effect of the Paris Agreement — a policy aimed at fostering the transition to a low-carbon economy — on European stock returns. Our empirical results suggest a statistically significant and negative short-run effect on the stock returns of firms that demonstrate a strong commitment to sustainability.

\end{abstract}

\vspace*{1.5em}



\vspace*{1.5em}



\newpage

\maketitle

\section{Introduction}
On 12 December 2015, during the 21st Conference of the Parties (COP21) to the United Nations Framework Convention on Climate Change (UNFCCC), world leaders of 196 countries endorsed the Paris Agreement (PA), a legally binding international treaty on climate change. This landmark agreement came as a surprise to the markets, defining a new course for global efforts to accelerate the transition to a low-carbon economy and mitigate the effects of climate change. Specifically, the PA sets out three ambitious goals: (i) "holding the increase in the global average temperature to well below 2°C above pre-industrial levels", (ii) "increasing the ability to adapt to the adverse impacts of climate change and foster climate resilience and low greenhouse gas emissions development", and (iii) "making finance flows consistent with a pathway towards low greenhouse gas emissions and climate-resilient development" (Art. 2, \cite{unfccc2015report}). In this context, financial markets play a central role in supporting the low-carbon transition by redirecting capital towards sustainable activities. By setting clear and credible long-term objectives, the Agreement has contributed to reshaping expectations and pricing mechanisms across financial markets, raising the question of how investors responded to its adoption. A relevant channel through which firms convey their environmental commitment is the issuance of green bonds, a tangible measure of the innovation driving sustainable finance in recent years. In this paper, we analyse the short-run response of European equity markets to the announcement of the Agreement, focusing on firms that signal environmental commitment through the issuance of green bonds. For this purpose, we use panel data on quarterly adjusted closing prices for the S\&P350 Europe index, which comprises 350 leading listed companies from 16 developed European countries. 

Recent econometric literature has shown a strong interest in policy evaluation methods for panel data as effective tools to adjust for unobserved confounding, that is a frequent challenge in economic studies. This area builds foundationally on approaches such as Difference-in-Differences (Diff-in-Diffs) 
and Two-Way Fixed Effects (TWFE) models, which are valid under the assumptions of parallel trends and causal effect homogeneity, respectively. However, growing interest in this field has led to numerous refinements and extensions that relax these often too restrictive assumptions. By moving beyond parallel trends, researchers can account for time-varying, unobserved confounders, and incorporating heterogeneous causal effects allows for a deeper understanding of variations in individual treatment responses.

Developments in causal panel factor models\footnote{Panel factor models also have a well-established use in non-causal applications (e.g., Chamberlain, 1984; Liang and Zeger, 1986; Arellano and Honoré, 2001), particularly in the modern literature (Bai, 2003, 2009; Pesaran, 2006; Moon and Weidner, 2015, 2017; Bai and Ng, 2017), even addressing complex asymptotic cases with increases in both time and cross-sectional dimensions.} are particularly promising for addressing the challenges of causal inference with longitudinal data due to their interpretability as interactive fixed effects models \citep{ArkhangelskyImbens2023, Xu2023}. Since the foundational work of \cite{Hsiao_et_al_2012}, the literature in this field, despite its methodological diversity, has been unified by a common thread: the reliance on assumptions of ignorability of the treatment assigmnent given a set of unobserved factor loadings in the alternative forms of strict exogeneity, parallel trends or sequential ignorability, even in conjunction with covariates \citep{Xu2023}.
Initially, this strand of the literature focused on inference methods under a finite-sample framework, in which causal effects are identified solely at the sample level. While some of the proposed methods are applicable only in settings with a single treated unit \citep[e.g.,][]{Amjad_et_al_2018, Arbour_et_al_2021, Hsiao_et_al_2012, Li_Bell_2017}, others extend the analysis to cases with 
multiple treated units \citep[e.g.,][]{abadie_et_al_2024, abadie_et_al_2023, arkhangelsky2021synthetic, athey2021matrix, AtheyImbens2025, Chernozhukov_et_al_2021, Fan_et_al_2022, Li_Sonnier_2023, Samartsidis_et_al_2020, Pang_et_al2022, xu2017generalized}. However, finite-sample inference has certain limitations due to its exclusive focus on the sample. On one hand, this approach ensures a high degree of internal validity, which is particularly valuable in experiments or comparative case analyses involving small target populations, potentially as small as a single treated unit. On the other hand, this narrow focus sacrifices external validity, especially when the policy being evaluated is intended for a larger, potentially time-varying population. Furthermore, when identification relies on unconfoundedness on observables or latent variables, as is common in observational studies, shifting from a finite-sample to a super-population perspective enables a more rigorous treatment of continuous covariates, considering them as draws from continuous distributions, which facilitates the inference \citep{Imbens_Rubin_2015}. Therefore another line of contributions in the literature on causal factor models has adopted a super-population approach. The earliest, chronologically, is the work by \cite{gobillon2016regional}, which proposes a comparison of alternative methods but does not provide formally justified inference. By contrast, formal inferential procedures are developed in more recent studies. For example, \citet{Imbens_et_al2021} who rely on time-invariant covariates as instruments, assuming that their effects on the outcome remain constant over time; \citet{CallawayKarami2023} who require post-treatment error terms to be independent of errors from previous periods; \cite{Feng_2021} who assumes the availability of a large set of noisy measurements related to the underlying latent confounders through an unknown, possibly nonlinear, factor structure; \cite{Deaner_et_al_2025} who impute untreated potential outcomes for treated units using outcomes from untreated units with similar values of the estimated latent confounders. Less restrictive identifying assumptions are proposed by \cite{Brown_Butts_2023} and \cite{Liu_et_al2024}, however, for the latter, at the cost of relying on computationally intensive methods, such as bootstrap and jackknife, for inference, which can represent a limitation in terms of ease of use for practitioners engaged in policy evaluation with panel data, particularly when working with datasets that include a significant number of treated units. While all of these contributions assume i.i.d. error terms, \cite{Chan_Kwok_2022} allow for a limited degree of cross-sectional correlation in the errors and, similarly to some of the aforementioned finite-sample contributions, consider both settings with a single treated unit and with multiple treated units. However, they identify both individual and average effects parametrically through regression models, which are subject to the well-known criticisms associated with parametric identification of causal effects. As emphasized by \cite{Imbens_Rubin_2015}, parametric identification relies critically on the correct parametric specification of the regression model, in addition to the fact that the estimand will generally be a more complex function of the parameter vector. Moreover, even under standard exogeneity assumptions on the error term of the parametric model, the parametric approach combines the unconfoundeness assumption with assumptions on functional form and effect homogeneity that are quite strong and often unnecessary.

In order to maintain weak identifying assumptions and allow for a limited degree of cross- and serial-correlation without incurring the aforementioned drawbacks typical of parametric causal identification, we propose a semi-parametric approach based on the use of the propensity score, defined as the probability of adopting the policy given a low-dimensional set of factor loadings, within a weighting framework.
This facilitates inference under the M-estimation theory, according to which the standard errors can then be derived in the form of a sandwich estimator. In this way, it is possible to take into account all sources of uncertainty: the drawn from a super-population, the estimation of factor loadings on which the propensity score will depend, the estimation of the policy model parameters, and the randomness of treatment assignment as specified by the policy model. Moreover, our approach, based on a low-dimensional set of factor loadings, offers a key empirical advantage. By exploiting the pre-treatment history of the outcome, it enables the implementation of inverse probability weighting in panel data settings even when at least one of the treatment groups has a sample size that is markedly smaller than the length of the pre-treatment time dimension. For the estimation and identification of factor loadings, we rely on the theory of \cite{bai2013principal}. Our contribution is then twofold. First, we propose a method for estimating causal effects in panel data that offers advantages in terms of causal interpretability, ease of use for policy makers and computational efficiency. Second, we evaluate the short-run effect of the Paris Agreement.

We find a statistically significant negative impact of the Agreement on stock returns for European firms that demonstrate a strong enviromental commitment to sustainability right after the policy implementation. This behavior reflects investors' willingness to accept lower expected financial performance, driving green firms to trade at a premium relative to brown counterparts due to their perceived lower risk and alignment with sustainability objectives. This result supports the strand of literature analysing the relation between stock returns and environmental performance. \citet{bolton2021investors,bolton2021carbon,bolton2023global} and \cite{bolton2022financial} find that investors demand a carbon risk premium as compensation for their exposure to transition risk, which became large and highly significant only after the Paris Agreement. A similar pattern is observed in the bond markets, where \cite{zerbib2019effect} documents the existence of a green premium with brown bonds trading at a discounted price relative to comparable green bonds. \cite{monasterolo2020blind} show that after the agreement markets adjusted the risk–return profile of low-carbon indices downward, reflecting a reduced risk perception of green assets. Our findings are also consistent with \cite{pastor2021sustainable} and \cite{pedersen2021responsible}, where it is shown that, in equilibrium, investor preferences for sustainability can result in lower expected returns for green assets relative to their brown counterparts.

The remaining sections of the paper are organized as follows. Section 2 provides an overview of the economic motivations behind the study and outlines the main institutional features of the Paris Agreement. Section 3 introduces the proposed methodology. Section 4 describes  a simulation study whereas Section 5 presents the empirical application. Section 6 concludes.

\section{The Paris Agreement}
The Paris Agreement, adopted on 12 December 2015 during COP21 in Paris, is a legally binding international treaty under the United Nations Framework Convention on Climate Change which sets out three overarching goals: limiting the increase in global average temperature to well below 2°C above pre-industrial levels, strengthening the ability to adapt to climate-related risks, and aligning financial flows with pathways consistent with low greenhouse gas emissions and climate-resilient development. Moreover, in contrast to the top-down approach of the Kyoto Protocol, the Paris Agreement introduces a more flexible and inclusive approach based on voluntary commitments where each country submits Nationally Determined Contributions (NDCs) which outline climate targets, policies, and measures and are subject to periodic review under an international transparency framework.

To achieve the goals of the Paris Agreement approximately \$93 trillion over the following 15 years are required (\citealp[]{oecd2017mobilising}), an amount of capital way too high to be only financed by the public sector and that call for the pivotal role of financial market participants in mobilizing investments for the global climate action (\citealp[]{lins2017social}). By supporting the transition to a low-carbon economy, the Paris Agreement increased awareness of climate-related risks (\citealt{bolton2021investors}; \citealp[]{fahmy2022rise}; \citealp[]{hsu2023pollution}). This has had two main effects. First, companies started to reassess their internal policies with a stronger focus on sustainability, increasingly redirecting resources toward green projects. This shift also led many firms to explore financing tools better aligned with their environmental strategies, such as green bonds (\citealp[]{flammer2021corporate}; \citealp[]{bhutta2022green}). Second, investor preferences have been redirected toward climate-conscious investments, incorporating climate risk into their financial strategies (\citealp[]{pedersen2021responsible}, \citealp[]{zerbib2022sustainable}). Investors increasingly seek to allocate capital to stocks of environmentally responsible firms, particularly those demonstrating a clear commitment to sustainability, thereby affecting their market value (\citealp[]{rohleder2022effects}; \citealp{de2023climate}).

In this context, green bonds have become a key signal through which firms communicate the commitment to environmental sustainability. Their structure makes them particularly well-suited to finance climate-related projects, which often involve high upfront costs, long maturities, and inflation-linked cash flows. Alongside these financial features, green bonds also provide a transparent framework that helps reinforcing the credibility of a firm’s environmental strategy. The proceeds are allocated exclusively to projects with clear environmental benefits and are subject to third-party verification to ensure that the funds are used as promised. Since the European Investment Bank issued the first green bond in 2007, the market has expanded rapidly, surpassing \$250 billion by 2019, with issuances increasing fivefold since 2016\footnote{Data are retrieved from https://www.climatebonds.net/resources/reports/green-bonds-global-state-market-2019. The site was accessed on March 12, 2025.}.

\section{The methodology}
\subsection{Setup and assumptions}

We proceed under the potential outcomes framework for causal inference \citep{Rubin1974, Rubin1978} and consider a sample of $N$ units, indexed by $i=1, \dots, N$, drawn from the super-population $\Omega$. These units are repeatedly observed at different times $t=1, \dots, T$, forming a panel dataset. Each unit has a pair of potential outcomes, $Y_{i,t}(Z_{i,t})$, where $Z_{i,t}$ denotes the binary treatment status, taking the value zero if unit $i$ is not treated at time $t$ and one otherwise. At the individual level, only one of the two potential outcomes is observed, specifically the one corresponding to the realized value of $Z_{i,t}$:  
\[
Y_{i,t}(Z_{i,t})= Y_{i,t}(1) \cdot \mathbbm{1}_{Z_{i,t}=1} + Y_{i,t}(0) \cdot \mathbbm{1}_{Z_{i,t}=0}.
\]

Since in our application the treatment is administered only at the last time $t=T$, only $Y_{i,t}(0)$ is observed for $t \leq T-1 = T_{0}$, while both $Y_{i,t}(0)$ and $Y_{i,t}(1)$ are potentially observed at time $t=T$. The potential outcomes approach allows us to define the population causal effect at a given time $t$ as the population mean of the individual differences between the two potential outcomes. Specifically, our goal is to quantify the population causal effect at time  $t=T$ for treated units:

\begin{equation}
	\tau^{ATT}= \mathbb{E}[Y_{i,T}(1) - Y_{i,T}(0) | Z_{i,T}=1].
	\label{eq: ATT_estimand}
\end{equation}

We outline our set of identifying assumptions.

\begin{assumption} \label{Assumption 1}
	Stable Unit Treatment Value Assumption.
\end{assumption}

This is a standard assumption in causal inference \citep{Rubin1978} by which: the potential outcome for one unit is not affected by the treatment assignment of other units, and the observed outcome for a unit equals the potential outcome for the treatment that the unit actually received. The assumption excludes interference between units (or spillover effects), as well different versions of a treatment.\\

\begin{assumption} \label{Assumption 2}
	Longitudinal factor model.
\end{assumption}

We assume an interactive fixed-effects specification for the outcome model in the pre-treatment period, $t \leq T_{0}$:
\[
Y_{i,t}(0) = \boldsymbol{\lambda}_i^{\prime} \mathbf{F}_t+\xi_{i,t},
\]
where $\mathbf{F}_t=\left(F_{1,t} , \ldots, F_{r ,t}\right)^{\prime}$ is a $r$-dimensional vector of time-varying factors with the corresponding loadings $\boldsymbol{\lambda}_i= \left(\lambda_{i,1}, \ldots, \lambda_{i,r}\right)^{\prime}$ for the $i$th unit, and $\xi_{i, t}$ represent idiosyncratic errors. The model can be arranged in matrix notation setting $\boldsymbol{Y(0)}$ for the $T_0 \times N$ matrix of the outcomes, $\boldsymbol{F}=\left(\boldsymbol{F}_1, \boldsymbol{F}_2, \ldots, \boldsymbol{F}_{T_0}\right)^{\prime}$, and $\boldsymbol{\Lambda}=\left(\boldsymbol{\lambda}_1 , \ldots, \boldsymbol{\lambda}_n\right)^{\prime}$.

\begin{assumption} \label{Assumption 3}
	(a) ${\mathbf{F}^{\prime} \mathbf{F}}/{T_{0}}=\mathbf{I}_r$ (b) $\boldsymbol{\Lambda}^{\prime} \boldsymbol{\Lambda}$ is a diagonal matrix with distinct entries.
\end{assumption}

This assumption is commonly used in the identification and estimation of factor models via the principal components method. It requires the diagonal elements of $\boldsymbol{\Lambda}^{\prime}\boldsymbol{\Lambda}$ to be distinct, positive, and arranged in decreasing order (\citealp[]{bai2013principal}; identification restriction PC1). This condition implies that the rotation matrix can be treated as the identity matrix in asymptotic representations, which ensures that the factors and loadings are uniquely determined from their product $\mathbf{F}\boldsymbol{\Lambda}^{\prime}$ thus eliminating rotational indeterminacy. This normalization, in turn, enables a consistent recovery of the latent factor structure using Principal Component Analysis (PCA).

\begin{assumption} \label{Assumption 4}
	Latent ignorability.
\end{assumption}

We assume that the source of unobserved confounding is entirely captured by the loadings $\boldsymbol{\lambda}_i$:

\[
\{Y_{i,T}(1), Y_{i,T}(0)\} \perp Z_{i,T} \ | \boldsymbol{\lambda}_i
\]

with $0 < P(Z_{i,T}|\boldsymbol{\lambda}_i) < 1$ which assures randomness in the treatment conditional on
the loadings. The idea of conditioning randomization on the loadings is common in causal models for panel data, as a way to correct for selection on unobservables\footnote{Alternative causal approaches for addressing selection on unobservables, such as Regression Discontinuity Design or the Instrumental Variable method, do not impose forms of randomization conditional on the loadings; however, they are significantly more limited in terms of external validity.} \citep{Xu2023}. 
In particular, compared to standard approaches like TWFE or DID, Assumption \ref{Assumption 4}, when combined with a factor-interactive fixed effects outcome model (Assumption \ref{Assumption 2}) does not impose parallel trends for the potential outcomes $Y_{i,t}(0)$. Latent ignorability is inherently untestable since it involves potential outcomes that can never be observed simultaneously. However, as with the classical unconfoundedness assumption, its plausibility can be assessed by examining the balance of the weighted distribution of the confounders (the loadings) across treatment groups.

\newpage

\begin{assumption}\label{Assumption 5}
	Propensity score specification.
\end{assumption}

The propensity score $e_i(\boldsymbol{\lambda}_i)
= \Pr(Z_{i,T}=1 \mid \boldsymbol{\lambda}_i)$
follows a logistic specification, i.e.
\[
e_i(\boldsymbol{\lambda}_i)
=
\frac{1}{1+\exp\!\left(-\boldsymbol{\lambda}_i^{\prime}\boldsymbol{\beta}\right)},
\]
for some $r$-dimensional parameter vector $\boldsymbol{\beta}$.
\vspace{0.35cm}{
	
	Assumptions \ref{Assumption 1}-\ref{Assumption 5}\footnote{Observed pre-treatment covariates could also be incorporated in the propensity score specification. In that case, the latent ignorability assumption would be weakened by requiring treatment assignment to be independent of potential outcomes conditional on both the latent factor loadings and the observed covariates. We abstract from this extension to maintain a fully latent-structure-based identification strategy. Since covariates are observed and of fixed dimension, their inclusion would not affect the asymptotic argument in Theorem 1, which can be straightforwardly extended.} describe a causal scenario where the treatment is assigned only once to each unit, ruling out staggered or sequential adoption of the treatment. Thus, the design follows a standard causal inference framework under treatment ignorability, for which several methods have been proposed in the literature including regression, matching, and propensity score-based approaches. Since our assignment rule is ignorable according to Assumption \ref{Assumption 4}, the aforementioned methods could be applied once the loadings are estimated. However, the main challenge remains quantifying the large-sample variances of these estimators, including the additional uncertainty arising from the estimation of the loadings. For this purpose, we propose to adopt an Inverse on the Propensity score Weighting (IPW hereafter) estimator\footnote{A Propensity Score matching estimator for a similar framework was introduced by \cite{gobillon2016regional}, but the authors did not provide a method for quantifying the variance.}, where the propensity score depends on the factor loadings, as its large-sample properties can be effectively derived using the M-estimation theory \citep{lunceford2004stratification}.}

\subsection{Estimation strategy}
We follow a three-step procedure. In the first step, we estimate the factor loadings, $\hat{\boldsymbol{\Lambda}}$, from the potential outcome model in the pre-treatment period using PCA. Under standard assumptions for approximate factor models (\citealp[]{bai2013principal}), we minimize the following  nonlinear least squares problem:

\begin{equation*}
	\begin{aligned}
		\min _{F, \Lambda} \sum_{i=1}^N \sum_{t=1}^{T_0}\left(Y_{i,t}(0)-\boldsymbol{\lambda}_i^{\prime} \boldsymbol{F}_t\right)^2 
	\end{aligned}     
\end{equation*}

subject to $\mathbf{F}^{\prime} \mathbf{F}/T_{0}=\mathbf{I}_r$ and $\boldsymbol{\Lambda}^{\prime} \boldsymbol{\Lambda}$  is a diagonal matrix with distinct entries. 

In the second step, we estimate the propensity score $e_i$ via logistic regression under Assumption~\ref{Assumption 5},
using the estimated factor loadings $\hat{\boldsymbol{\lambda}}_i$ obtained via PCA. The estimated propensity score is given by

\begin{equation*}
	\begin{aligned}
		e_i(\hat{\boldsymbol{\lambda}}_i) = \frac{1}{1 + \exp(-\hat{\boldsymbol{\lambda}}_i' \boldsymbol{\beta})},
	\end{aligned}
\end{equation*}

where the parameter vector \( \boldsymbol{\beta} \) of dimension \( r \times 1 \) is obtained through Maximum Likelihood Estimation (MLE). 

In the third step, we compute the average treatment effect on the treated (ATT). To estimate the causal effect on the treated population, we employ the H\'{a}jek estimator based on propensity score weighting, which helps reduce the variance and provides stable estimates (\citealp[]{hirano2003efficient}). Specifically, the ATT is computed as:

\begin{equation}
	\begin{aligned}
		\hat{\tau}^{ATT} = \hat{\tau}_1 - \hat{\tau}_0 = \frac{\sum_{i=1}^{N}Z_{i,T}\,Y_{i,T}}{\sum_{i=1}^{N}Z_{i,T}} - \frac{\sum_{i=1}^{N} \,(1-Z_{i,T})\,Y_{i,T}\,(\hat{e}_{i}/(1-\hat{e}_{i}))}{\sum_{i=1}^{N} \,(1-Z_{i,T})(\hat{e}_{i}/(1-\hat{e}_{i}))}.
	\end{aligned}
	\label{eq: ATT}
\end{equation}

Under a super-population framework, the ATT is defined as the population quantity (\ref{eq: ATT_estimand}), where expectations are taken with respect to the joint distribution of potential outcomes, treatment assignment, and latent factor loadings. Under Assumption \ref{Assumption 4} (latent ignorability) and the usual overlap condition between treatment status for the propensity score, conditioning on the loadings removes selection bias, so that the counterfactual outcome mean for treated units can be recovered from appropriately reweighted outcomes of untreated units using the propensity score \citep{Ding_Li_2018}. Equation (\ref{eq: ATT}) is the H\'{a}jek sample analogue of this population representation which identifies the super-population ATT.

In Theorem \ref{thm1} we derive a closed-form expression for the asymptotic variance of the ATT estimator through M-estimation theory. This approach provides a unified framework that accounts for the estimation uncertainty originating from the PCA-based factor loadings, the logistic regression estimates, and the Hájek weighting scheme. In particular, it delivers robust standard errors and confidence intervals for the ATT. Moreover, the M-estimation framework provides a direct and efficient solution for inference, while avoiding computationally intensive resampling methods such as bootstrap and jackknife. Theorem \ref{thm1} also shows that the asymptotic variance can be represented, via an influence function, as the sum of three components capturing the sampling variability of the treated mean, the reweighted control mean, and the propensity score, while the effect of estimating the factor loadings is asymptotically negligible.

\begin{theorem}\label{thm1}
	Under Assumptions \ref{Assumption 1}-\ref{Assumption 4} and considering standard assumptions for approximate
	factor models (\citealp[]{bai2013principal}), assuming a logistic propensity score, for $N, T_{0} \rightarrow \infty$ with $\sqrt{T_0}/N \rightarrow 0$,
	
	\begin{equation}
		\label{eq:variance_ATT}
		\sqrt{N}\left(\hat{\tau}^{ATT}-\tau^{ATT}\right)
		\xrightarrow{d} \mathcal{N}\bigl(0,\,\mathbb{E}\bigl[I_i^{2}\bigr] \bigr), 
	\end{equation}
	
	where 
	
	\[
	I_i
	=
	\eta_{1}^{-1}\,U_i(\tau_1)
	-
	\eta_{2}^{-1}\,U_i(\tau_0,\boldsymbol{\beta})
	-
	\eta_{2}^{-1}\,\boldsymbol{H}_\beta^{\top}\,\boldsymbol{E}_{\beta \beta}^{-1}
	\,\boldsymbol{S}_i(\boldsymbol{\beta},\boldsymbol{\lambda}_i),
	\]
	
	and
	
	\[
	\eta_{1} = \mathbb{E}\bigl[Z_{i,T}\bigr],\qquad
	\eta_{2} = \mathbb{E}\Bigl[(1-Z_{i,T})\,\frac{e_{i}}{1-e_{i}}\Bigr],
	\]
	\[
	U_i(\tau_1) = Z_{i,T}(Y_{i,T}-\tau_1),\qquad
	U_i(\tau_0,\boldsymbol{\beta})
	= (1-Z_{i,T})\,\frac{e_i}{1-e_i}\,(Y_{i,T}-\tau_0),
	\]
	\[
	e_i = e(\boldsymbol{\lambda}_i,\boldsymbol{\beta})
	= \frac{\exp(\boldsymbol{\lambda}_i^{\top}\boldsymbol{\beta})}
	{1+\exp(\boldsymbol{\lambda}_i^{\top}\boldsymbol{\beta})},
	\]
	\[
	\boldsymbol{H}_\beta
	=
	\mathbb{E}\Bigl[(1-Z_{i,T})\,\frac{\partial e_{i}}{\partial \boldsymbol{\beta}}\,
	\frac{1}{(1-e_i)^2}\,\bigl(Y_{i,T}-\tau_0\bigr)\Bigr],
	\]
	
	with $\boldsymbol{S}_i\left(\boldsymbol{\beta},\boldsymbol{\lambda}_i\right)$
	denoting the score of the propensity score model, while 
	$\boldsymbol{E}_{\beta \beta} 
	= {-}\,\mathbb{E}\!\left[\frac{\partial \boldsymbol{S}_{i}
		\left(\boldsymbol{\beta},\boldsymbol{\lambda}_i \right)}{\partial \boldsymbol{\beta}}\right]$ is the information matrice for the propensity score.
	
\end{theorem}

A formal proof is given in the Appendix \ref{sec12}.

\begin{remark}[Contribution of estimated factor loadings]
	\label{rem:loadings} Lemma~\ref{lemma1} and the asymptotic results of \citet{bai2013principal} imply that the additional variance component generated by the estimated factor loadings in the propensity score is asymptotically negligible. Hence, the asymptotic variance $\mathbb{V}^{ATT}$ in \eqref{eq:variance_ATT} coincides with the variance that would obtain if the factor loadings $\{\boldsymbol{\lambda}_i\}$ were observed.
\end{remark}

\begin{remark}[Feasible variance estimation]
	Appendix~\ref{app: empirical} describes how to implement the variance formula in Theorem~\ref{thm1} in practice. In particular, we construct a plug-in estimator of $\mathbb{V}^{ATT}$ by evaluating $I_i$ at the sample analogues of all population quantities.
\end{remark}

\section{Simulation study}
This section illustrates the performance of the proposed estimator under a range of simulated scenarios. We begin with a baseline design featuring i.i.d. idiosyncratic errors, which is analysed under alternative configurations of the cross-sectional dimension $N$ and the time dimension $T$. We then turn to a setting with a misspecified propensity score and finally consider scenarios in which the idiosyncratic component exhibits weak serial and cross-sectional dependence.

\subsection{Baseline Simulation Design}

We first consider two simulation frameworks based on a panel dataset with $N=500$ units and $T=100$ time points. The policy intervention is implemented at the final time period $T$. For both simulation frameworks, we set the number of time-varying factors equal to $r=3$, and simulate $\mathbf{F}_t$ from an AR(1) specification $\mathbf{F}_t = \mathbf{\Phi} \mathbf{F}_{t-1}+ \boldsymbol{\eta}_t$, where $\mathbf{\Phi}$ is a $3 \times 3$ diagonal matrix with the entries equal to $0.5$. In order to satisfy Assumption \ref{Assumption 3}, we simulate the $N \times 3$ loading matrix
\[
\boldsymbol{\Lambda} =
\begin{bmatrix}
	\boldsymbol{\Lambda}_1 & 0 & 0 \\
	0 & \boldsymbol{\Lambda}_2 & 0 \\
	0 & 0 & \boldsymbol{\Lambda}_3 \\
\end{bmatrix}.
\]

The first factor affects only the first $N_1$ units by the loadings $\boldsymbol{\Lambda}_1$, the second factor only the next $N_2$ units by the loadings $\boldsymbol{\Lambda}_2$, and the third factor the remaining $N_3=N-N_1-N_2$ by the loadings $\boldsymbol{\Lambda}_3$.
In the two frameworks, each of the three diagonal entries of $\boldsymbol{\Lambda}'\boldsymbol{\Lambda}$ is either less (Case1) or strictly greater (Case2) than one. For each simulation, we determine the sample sizes $N_1$, $N_2$, and $N_3$ by drawing from a trinomial distribution with parameters $p_1 = p_2 = \frac{1}{3}$. The loadings are drawn from independent standard normal distributions, $N(0,1)$, and the resulting $N \times 3$ matrix is then multiplied by a $3 \times 3$ diagonal matrix whose diagonal entries are set to $(1, 0.875, 0.75)$ for Case 1 and $(2.25, 2, 1.75)$ for Case 2. This ensures that Case 2 exhibits greater variability in the loading distributions. 
The treatment assignment indicator, $Z_i$, is generated from a logistic model $P(Z_i=1|\boldsymbol{\lambda}_i)=\exp(\boldsymbol{\lambda}_i'\boldsymbol{\beta})/(1+\exp(\boldsymbol{\lambda}_i'\boldsymbol{\beta}))$.

We simulate $Y_{i,t}(0)$ for all units at all time points $t=1,..,T$ from the factor model $Y_{i,t}(0) = \boldsymbol{\lambda}_i^{\prime} \mathbf{F}_t+\xi_{i,t}$ where $\xi_{i,t} \stackrel{\text{i.i.d.}}{\sim} N(0,1)$. For treated units at time $T=100$, the outcome $Y_{i,T}(1)$ is simulated as $Y_{i,T}(1) = \boldsymbol{\lambda}_i^{\prime} \mathbf{F}_{T}+\xi_{i,T}+\tau_{i}$, where $\tau_{i}=\tau_{}^{ATT}+u_i$, with $\tau_{}^{ATT}$ set to $2$ in each simulation settings, and $u_{i} \stackrel{\text{i.i.d.}}{\sim} N(0,1)$. 

For each of the two simulation cases, we consider four alternative scenarios in which the coefficient vectors $\boldsymbol{\beta}$ for the propensity score are selected to achieve varying degrees of balance in the loading distributions between treated and control groups, while ensuring that the proportion of treated units remains within the range of $15\%–20\%$. The values of $\boldsymbol{\beta}$ are reported in the second column of Tables \ref{Res_Sim_ATT_extended} and \ref{Res_Sim_beta_loadings_oracle}. Figure \ref{fig:scenarios_1} presents boxplots of the Absolute Standardized Difference (ASD) of the loadings between treated and control units, based on 1,000 simulated samples, for both the original  (unweighted) and weighted (on the propensity score) samples under Case 1. ASD is the absolute difference in the means of the weighted loading between the treatment and control groups divided by square root of the sum of within group variances:

\begin{equation*}
	\text{ASD} = {\left|\frac{\sum_{i=1}^N \lambda_i Z_i w_i}{\sum_{i=1}^N Z_i w_i} - \frac{\sum_{i=1}^N \lambda_i (1- Z_i) w_i}{\sum_{i=1}^N (1-Z_i) w_i}\right|}\Bigg /{\sqrt{s_{1}^2/N_1 + s_{0}^2/N_0}},
	\label{eq:std_bias}
\end{equation*}

where $N_z$ is the number of units and $s_z^2$ is the variance of the unweighted loading in group $Z=z$ for $z=0,1$. For the original data, $w_i=1$ for each unit and ASD is the standard two-sample t-statistic; for the weighted samples, $\omega_i$ are the ATT weights: $w_i=1$ for $Z_i=1$,  $w_i=e_i/(1-e_i)$ for $Z_i=0$. The figure highlights a high degree of imbalance in the loadings between treated and control groups for the original samples in the first scenario. This imbalance progressively decreases across the scenarios, reaching a good level of balance for all three loadings in the fourth scenario. In contrast, the boxplots for the weighted samples demonstrate a significantly higher degree of balance, with the average ASD consistently remaining below the threshold of 1.96 (the critical value of the two-sample \textit{t}-statistic at 0.05 level). The corresponding boxplots for Case 2 are shown in Figure \ref{fig:scenarios_2}. Due to the higher variability in the loading distributions, the original samples exhibit worse balance in all situations, while the weighted samples show an imbalance only in the first scenario, where the ASD for $\lambda_3$ exceeds the critical threshold of 1.96. 
\begin{figure}[ht!]
	\centering
	\subfloat
	{\includegraphics[height=5cm, width=7cm]{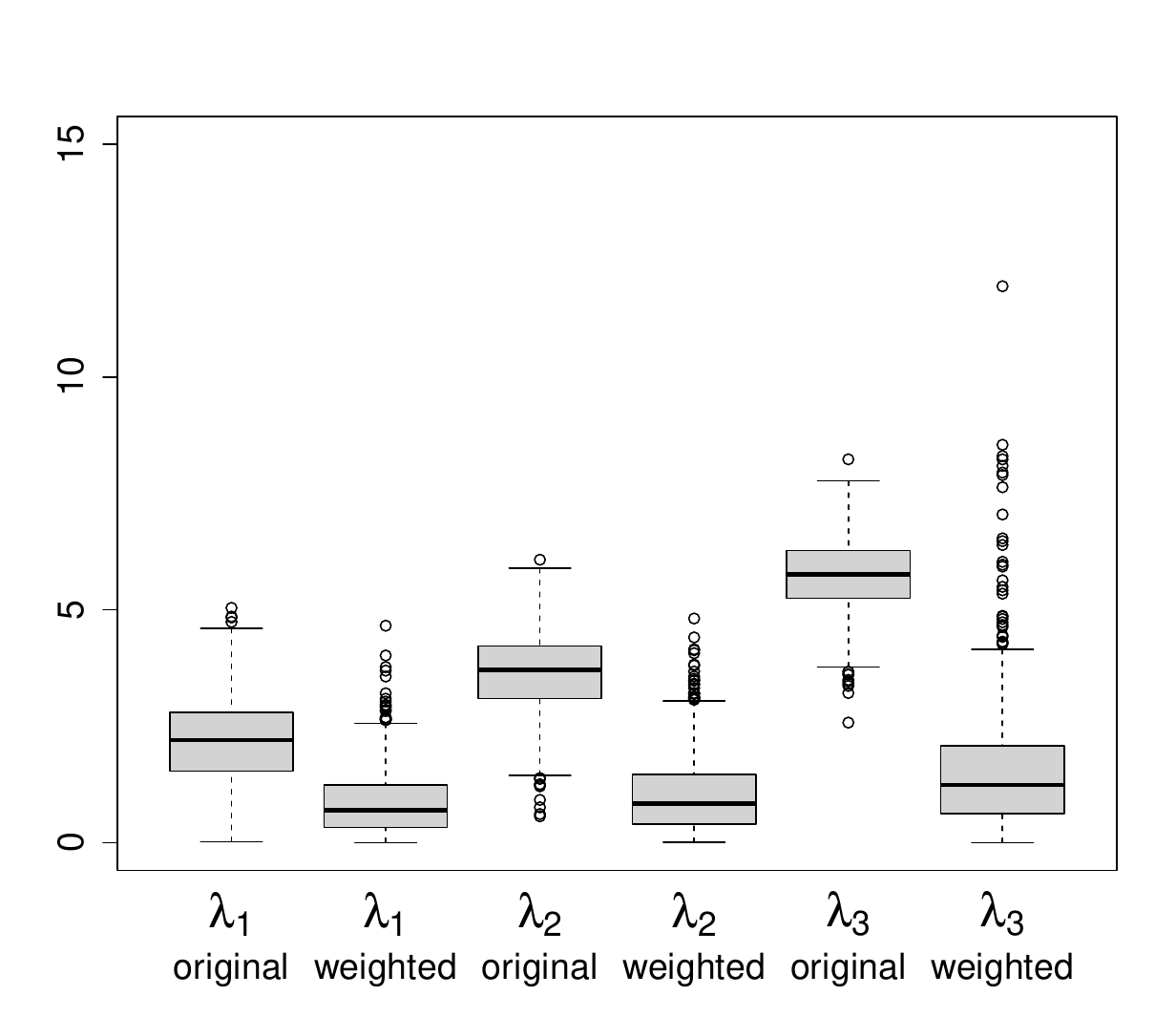}\label{Dens-11}}
	\subfloat
	{\includegraphics[height=5cm, width=7cm]{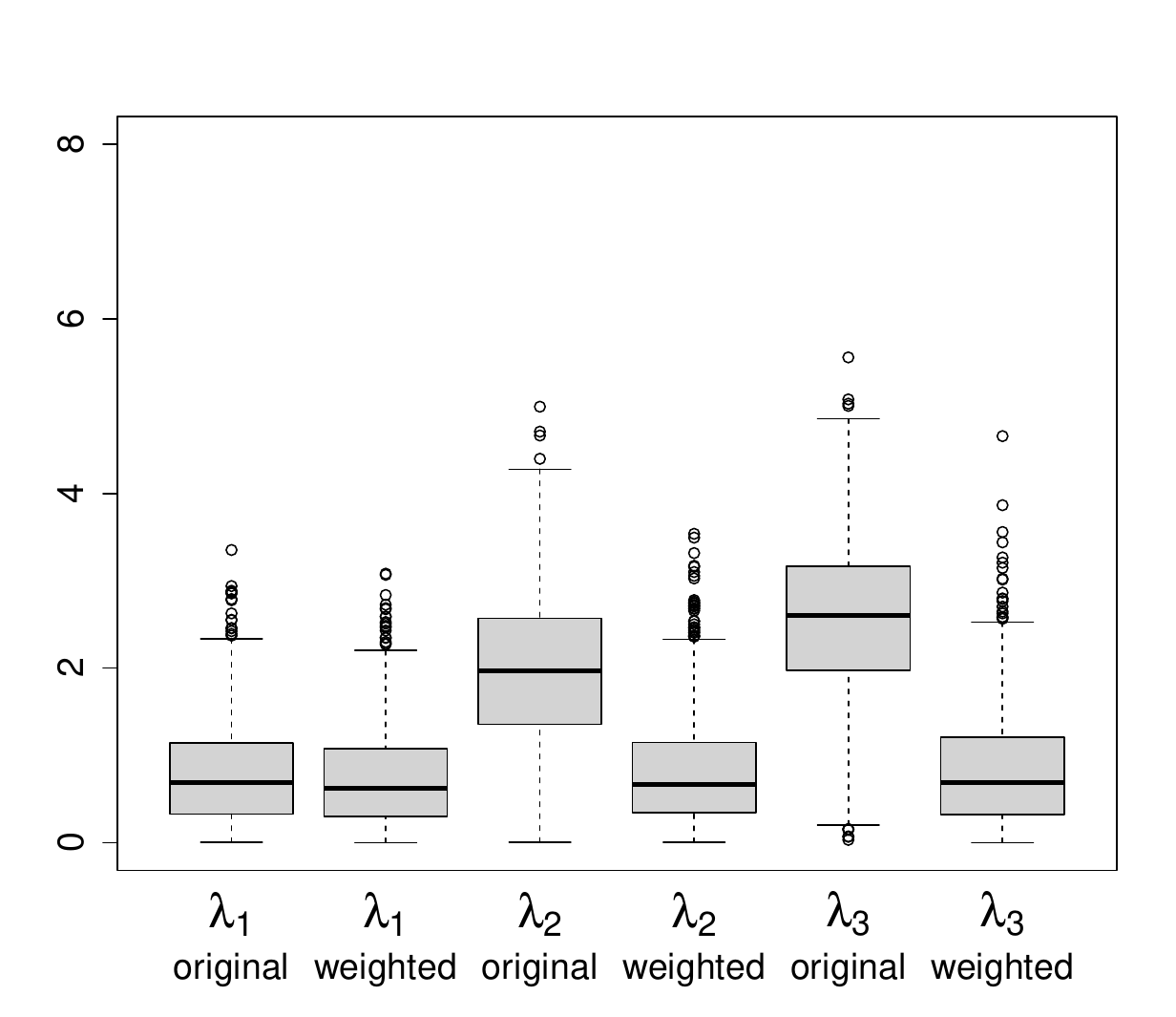}\label{Dens-12}}
	\par
	\subfloat
	{\includegraphics[height=5cm, width=7cm]{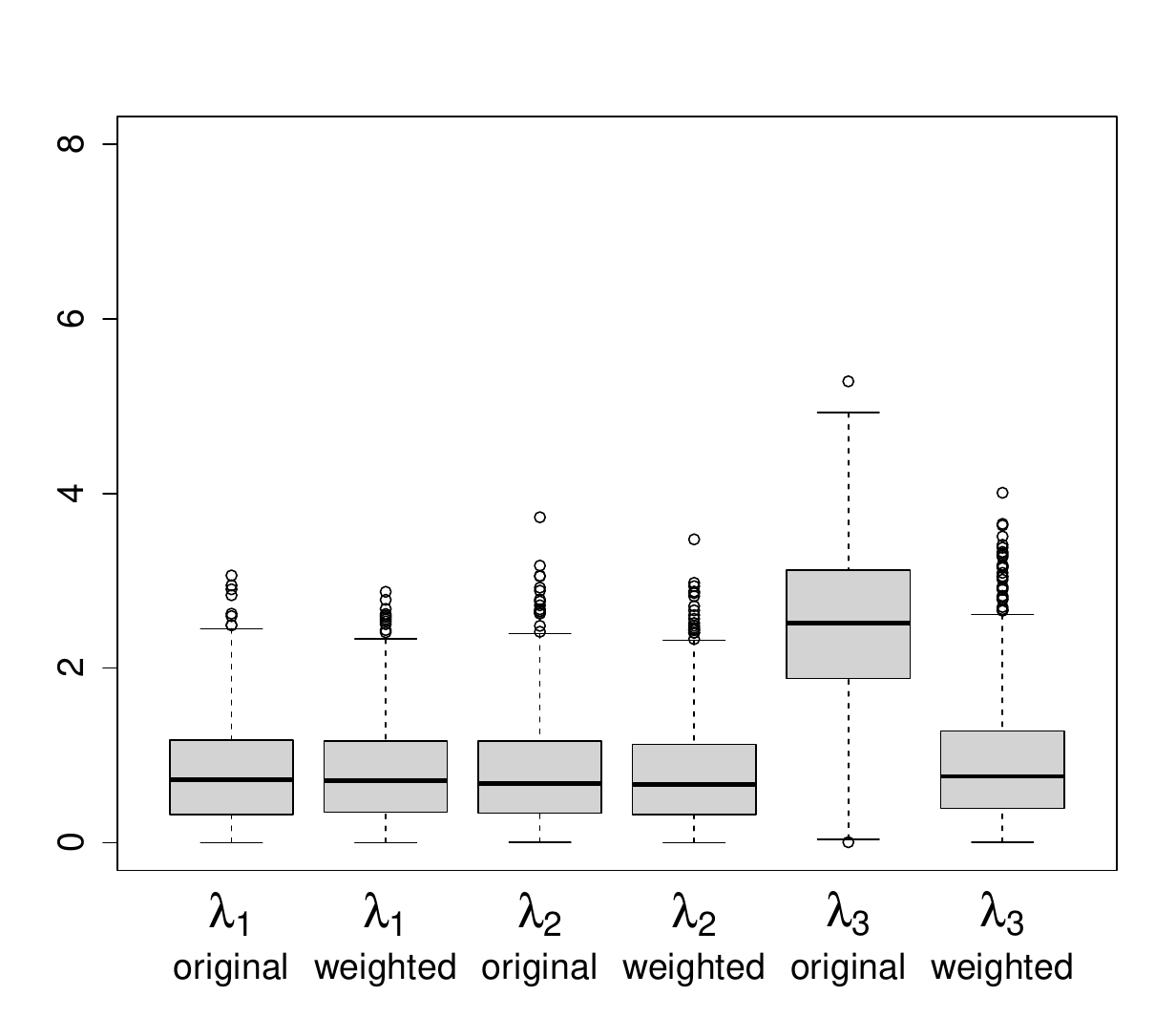}\label{Dens-13}}
	\subfloat
	{\includegraphics[height=5cm, width=7cm]{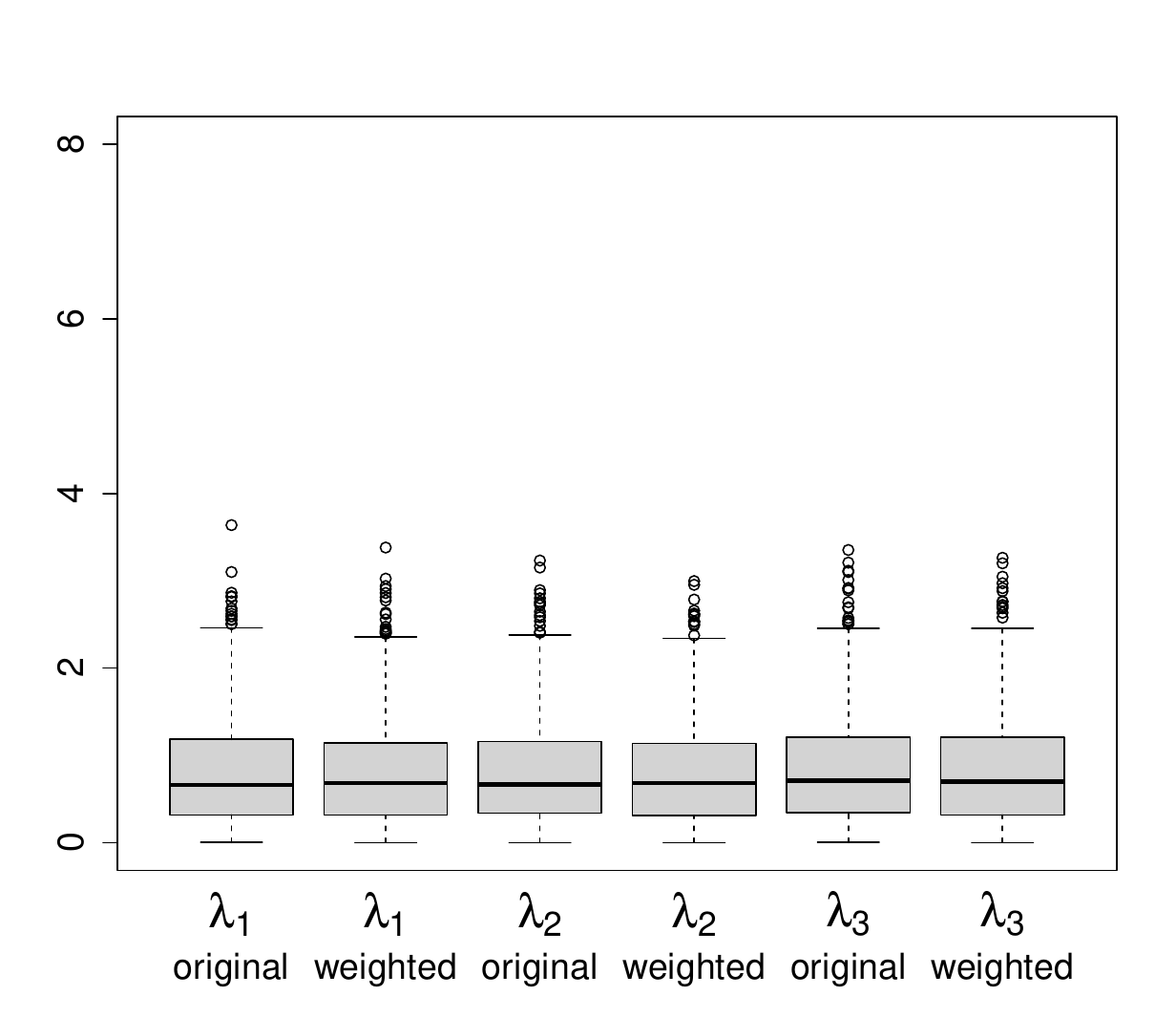}\label{Dens-31}}
	\caption{Boxplots of the absolute standardized difference, between treated and control units, for the original and the weighted samples for each loading, for the proposed scenarios from 1000 simulated samples; Case1.}
	\label{fig:scenarios_1}
\end{figure}

\begin{figure}[ht!]
	\centering
	\subfloat
	{\includegraphics[height=5cm, width=7cm]{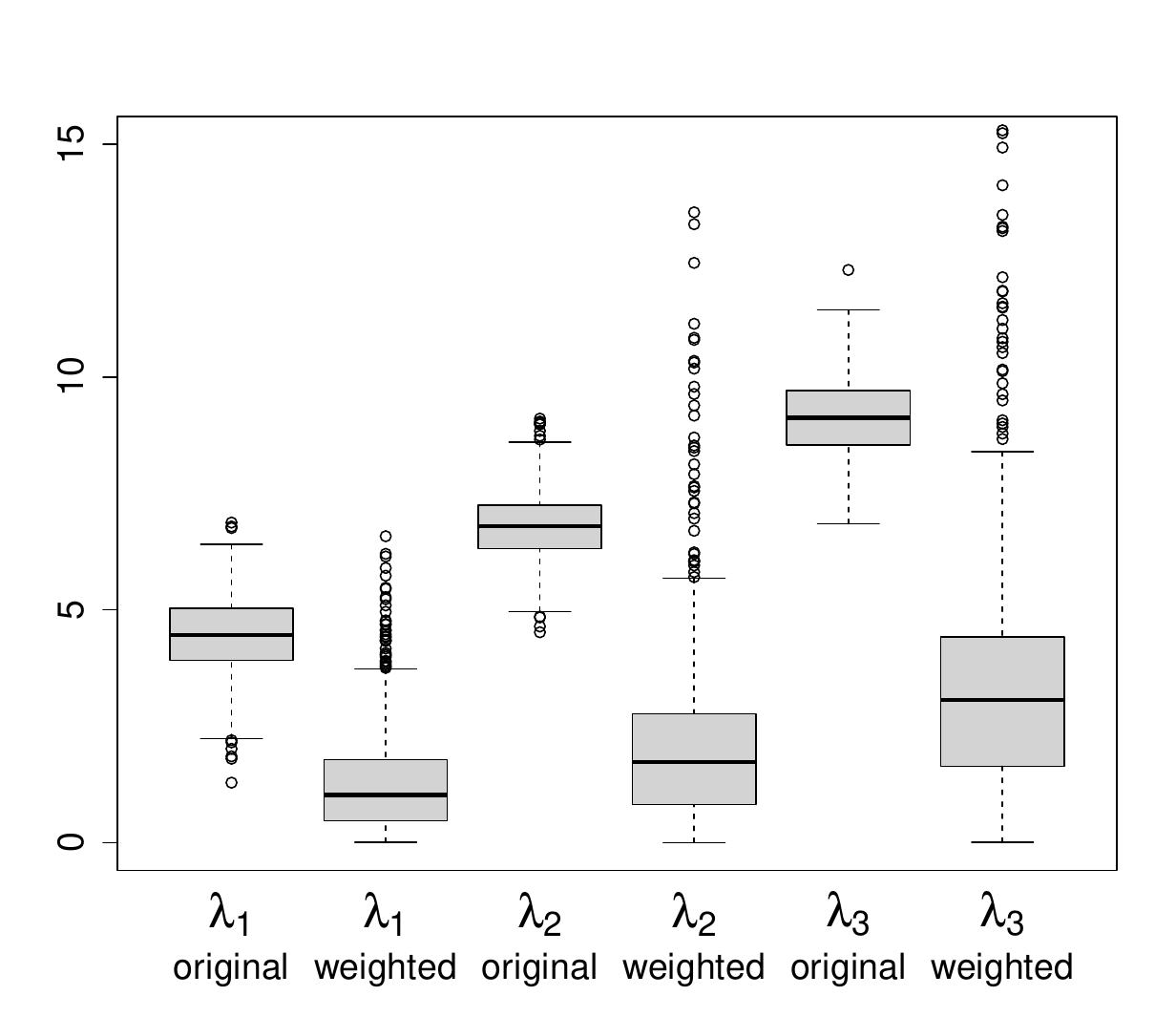}\label{Dens-11}}
	\subfloat
	{\includegraphics[height=5cm, width=7cm]{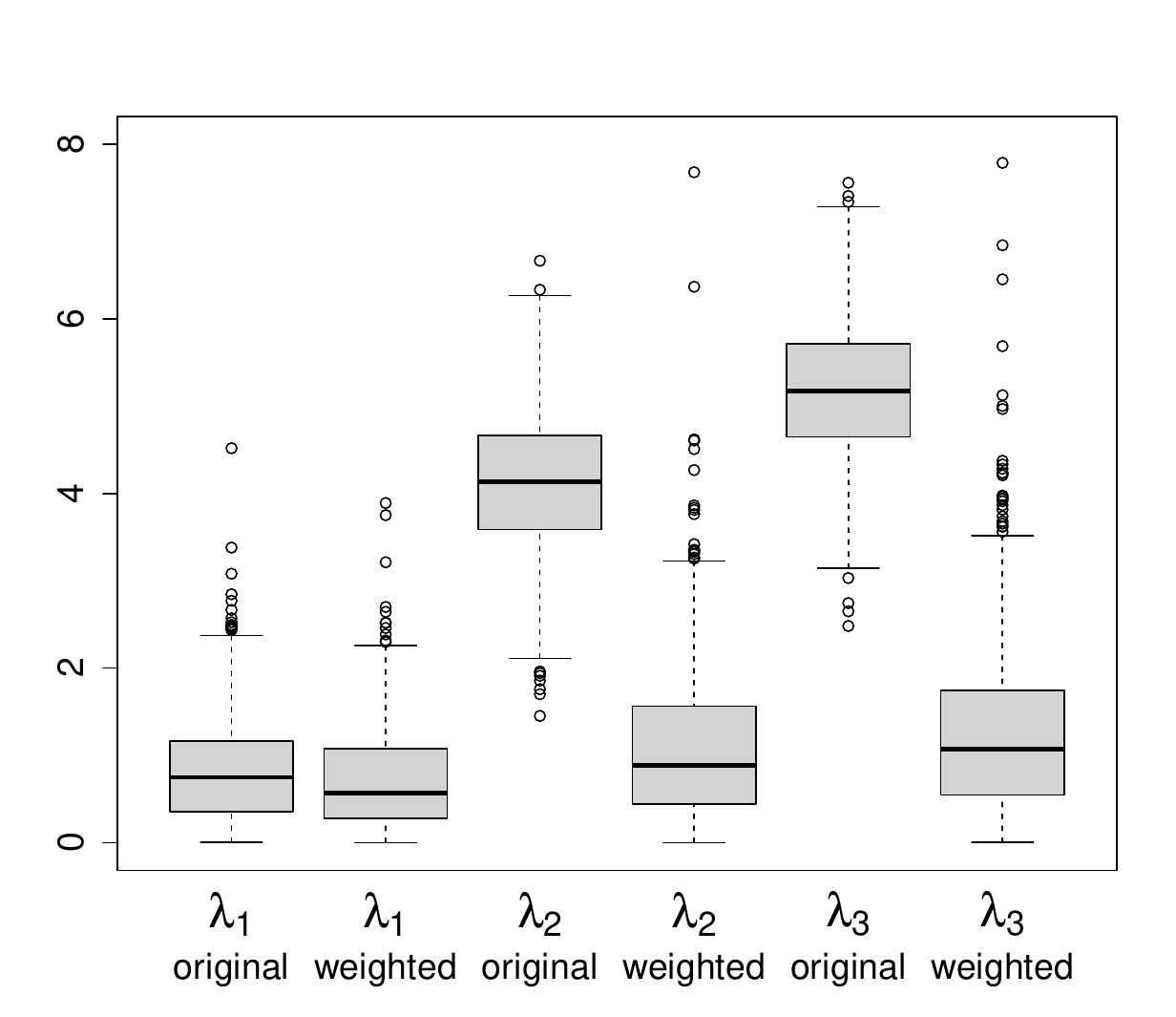}\label{Dens-12}}
	\par
	\subfloat
	{\includegraphics[height=5cm, width=7cm]{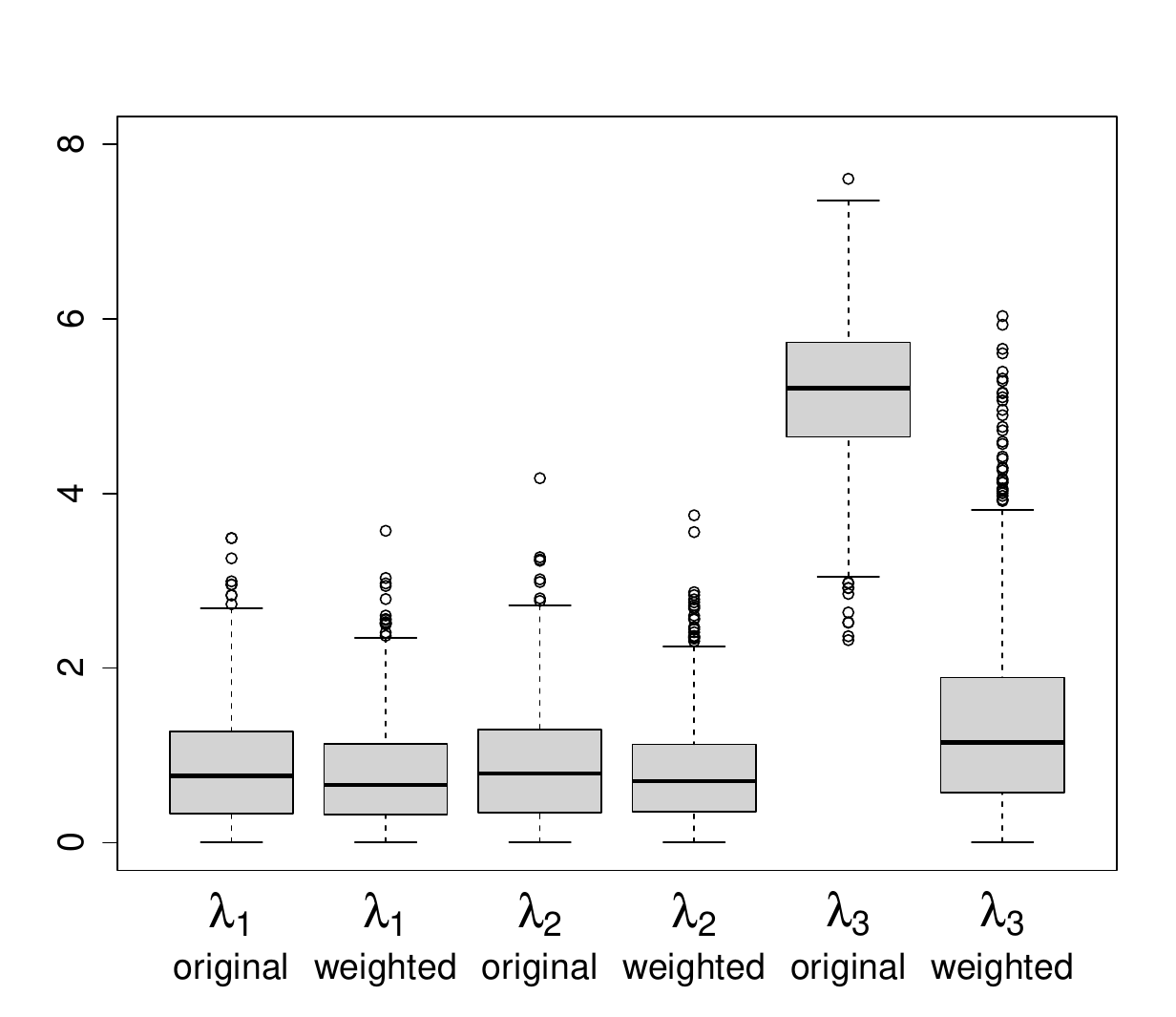}\label{Dens-13}}
	\subfloat
	{\includegraphics[height=5cm, width=7cm]{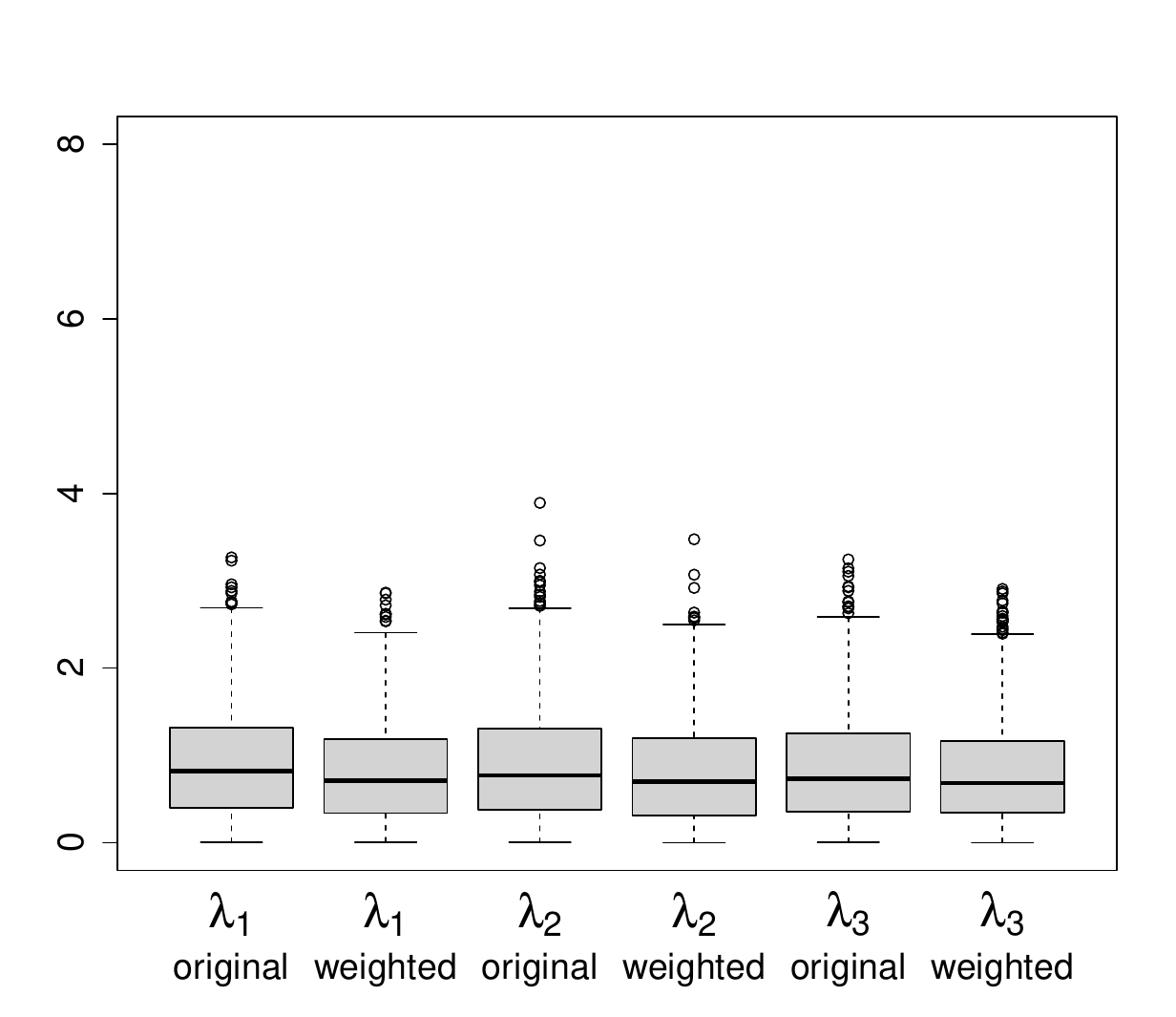}\label{Dens-31}} 
	\caption{Boxplots of the absolute standardized difference, between treated and control units, for the original and the weighted samples for each loading, for the proposed scenarios from 1000 simulated samples; Case2.}
	\label{fig:scenarios_2}
\end{figure}


To assess the performance of the proposed Inverse on the Propernsity score Weighting estimator (IPW), Table \ref{Res_Sim_ATT_extended} reports the Root Mean Squared Error (RMSE) computed with respect to the true value of the ATT, $\tau^{ATT}_{100}$, which is set to 2 at time $T=100$ across all estimation scenarios, $
\text{RMSE}=\sqrt{\frac{1}{N_{rep}} \sum_{k=1}^{N_{rep}}(\hat{\tau}^{ATT}_{100,k}-\tau^{ATT}_{100})^2}$,
where $\hat{\tau}^{ATT}_{100,k}$ denotes the estimated ATT in replication $k$, and the number of replications $N_{rep}=1000$.
To evaluate the calibration of the point estimates, we also report the coverage of the $95\%$ confidence intervals (CI)\footnote{Coverage is computed as the fraction of replications in which the nominal $95\%$ Wald interval
	$\hat{\tau}^{ATT}_{100,k}\pm1.96\,\widehat{\mathrm{se}}(\hat{\tau}^{ATT}_{100,k})$
	contains the true value $\tau^{ATT}_{100}$, where $\widehat{\mathrm{se}}(\hat{\tau}^{ATT}_{100,k})$ is obtained from the feasible plug-in variance estimator implied by Theorem~\ref{thm1}; see Appendix~\ref{app: empirical}.}. The IPW approach achieves low RMSE values and nominal $95\%$ CI coverage (i.e., around 0.95) in most of the scenarios of both cases, except for the first scenario of Case 2 due to the higher variability of the estimated weights induced by the greater dispersion of the factor loadings, which in turn leads to propensity scores that are more frequently close to the boundaries of the unit interval, thereby increasing the variability of the IPW estimator. We compare our results with those obtained using the Generalized Synthetic Control (GSC) method \citep{xu2017generalized} and the Synthetic Difference-in-Differences (SDID) method \citep{arkhangelsky2021synthetic}, which have gained widespread popularity in applications of causal inference with panel data for policy evaluation. The RMSE values of the proposed IPW estimator are generally higher than those obtained with GSC, whereas they are comparable to those of SDID in Case 1 and become larger in Case 2, particularly under stronger imbalance. The comparatively larger RMSE of IPW in some settings should also be interpreted in light of the different causal estimands, as our methodology accounts for the randomness introduced by sampling from the super-population, whereas GSC and SDID target finite-sample causal effects\footnote{As explicitly stated in \cite{xu2017generalized}, his analysis adopts a finite-sample causal inference framework. The estimand is the ATT for the units actually observed in the sample, which are treated as fixed rather than as a random draw from a super-population. A finite-sample perspective is also adopted by SDID \citep{arkhangelsky2021synthetic}, whose target estimand is the average treatment effect for the observed treated units over the observed post-treatment periods. The two methods differ, however, in their inferential procedures. GSC relies on both parametric and nonparamtric bootstrap, whereas \cite{arkhangelsky2021synthetic} propose clustered bootstrap, jackknife, and placebo-based variance estimators for SDID.}. We also consider an Oracle version of the IPW estimator, $\text{IPW}^{\text{Oracle}}$, in which the true factor loadings, known by construction in the simulation design, are used in place of their estimates. This comparison provides a direct assessment of the effect of first-step estimation error and confirms the theoretical implication (see Remark \ref{rem:loadings}) that the effect of estimating the factor loadings is asymptotically negligible for $\mathbb{V}^{ATT}$.

Table \ref{Res_Sim_beta_loadings_oracle} presents a single  RMSE computed jointly for the three loadings, as well as the RMSE for each estimated coefficient $\beta_j$ of the propensity score, respectively:

\[
\ \ \ \sqrt{\frac{1}{N_{rep}} \sum_{k=1}^{N_{rep}} \sum_{i=1}^{N}\sum_{r=1}^{3} (\hat{\lambda}_{r,i,k}-\lambda_{r,i})^2};\;\;\; \sqrt{\frac{1}{N_{rep}} \sum_{k=1}^{N_{rep}}(\hat{\beta}_{j,k}-\beta_j)^2}, \ j:0,1,2,3. 
\]

The RMSE of the estimated loadings is larger in Case 2 due to their higher variability. However, this greater dispersion increases the information content of the logistic regression in the propensity score model, leading to similar finite-sample performance of the $\beta_s$ estimates across the two cases. As in Table \ref{Res_Sim_ATT_extended}, we also report the RMSEs from the $\text{IPW}^{\text{Oracle}}$ estimator for comparison. In scenarios with larger (in absolute value) coefficients $\beta_s$ and in Case 2, where the factor loadings are more dispersed, the RMSE is higher when estimated loadings are used. However, the additional first-step variability affects the nuisance parameters, $\beta_s$, only and does not enter the asymptotic variance $\mathbb{V}^{ATT}$, given the nonlinear second-step mapping based on normalized propensity score weights. Accordingly, differences in the precision of the $\beta_s$ do not translate into comparable differences in ATT performance (see Table \ref{Res_Sim_ATT_extended}), consistently with Theorem 1 and the semi-parametric two-step structure of the estimator.

\begin{table}[!th]
	\scriptsize
	\centering
	\renewcommand{\arraystretch}{1.20}
	\resizebox{\textwidth}{!}{%
		\begin{tabular}
			{c|l| c c | c c | c c | c c}
			\hline
			& &
			\multicolumn{2}{c|}{IPW}
			& \multicolumn{2}{c|}{IPW$^{\mathrm{Oracle}}$}
			& \multicolumn{2}{c|}{GSC}
			& \multicolumn{2}{c}{SDID} \\
			\hline
			\text{Case} & \text{Scenario}
			& RMSE & Coverage
			& RMSE & Coverage
			& RMSE & Coverage
			& RMSE & Coverage \\
			& &
			& of 95\% CI
			& & of 95\% CI
			& & of 95\% CI
			& & of 95\% CI \\
			\hline
			& $\beta_0=-1.75, \beta_1=0.5, \beta_2=1, \beta_3=2$
			& 0.274 & 0.937
			& 0.273 & 0.946
			& 0.229 & 0.945
			& 0.281 & 0.898 \\
			1
			& $\beta_0=-1.75, \beta_1=0.05, \beta_2=0.5, \beta_3=0.75$
			& 0.245 & 0.945
			& 0.240 & 0.949
			& 0.230 & 0.950
			& 0.252 & 0.954 \\
			& $\beta_0=-1.75, \beta_1=0.05, \beta_2=0.05, \beta_3=0.75$
			& 0.260 & 0.936
			& 0.256 & 0.943
			& 0.240 & 0.945
			& 0.266 & 0.948 \\
			& $\beta_0=-1.75, \beta_1=0.05, \beta_2=0.05, \beta_3=0.05$
			& 0.252 & 0.950
			& 0.249 & 0.953
			& 0.232 & 0.958
			& 0.261 & 0.951 \\
			\hline
			& $\beta_0=-1.75, \beta_1=0.5, \beta_2=1, \beta_3=2$
			& 0.706 & 0.763
			& 0.689 & 0.794
			& 0.238 & 0.911
			& 0.354 & 0.835 \\
			2
			& $\beta_0=-1.75, \beta_1=0.05, \beta_2=0.5, \beta_3=0.75$
			& 0.313 & 0.942
			& 0.310 & 0.952
			& 0.229 & 0.942
			& 0.273 & 0.933 \\
			& $\beta_0=-1.75, \beta_1=0.05, \beta_2=0.05, \beta_3=0.75$
			& 0.311 & 0.933
			& 0.312 & 0.938
			& 0.231 & 0.938
			& 0.273 & 0.941 \\
			& $\beta_0=-1.75, \beta_1=0.05, \beta_2=0.05, \beta_3=0.05$
			& 0.250 & 0.945
			& 0.247 & 0.950
			& 0.230 & 0.957
			& 0.258 & 0.957 \\
			\hline
		\end{tabular}%
	}
	\caption{Root Mean Squared Errors (RMSE) and coverage of 95\% confidence intervals for the ATT from the proposed IPW estimator, the Oracle IPW estimator (based on true factor loadings), the Generalized Synthetic Control method (GSC), and Synthetic Difference-in-Differences (SDID), corresponding to different balancing scenarios, for Case 1 and Case 2. In all simulation cases, the panel dimensions are fixed at $N=500$ and $T=100$.}
	\label{Res_Sim_ATT_extended}
\end{table}

\begin{table}[!th]
	\scriptsize{
		\begin{center}
			\setlength{\tabcolsep}{2.1pt}
			\renewcommand{\arraystretch}{1.1}
			\begin{tabular}
				{c|l| c | c c c c | c c c c}
				\hline
				&  &  & \multicolumn{4}{c|}{IPW} & \multicolumn{4}{c}{IPW$^{Oracle}$} \\
				\hline
				\text{Case} & \text{Scenario} & Loadings
				& $\beta_0$ & $\beta_1$ & $\beta_2$ & $\beta_3$
				& $\beta_0$ & $\beta_1$ & $\beta_2$ & $\beta_3$ \\
				\hline
				& $\beta_0=-1.75, \beta_1=0.5, \beta_2=1, \beta_3=2$
				& 0.322
				& 0.142 & 0.514 & 0.658 & 0.792
				& 0.144 & 0.222 & 0.273 & 0.343 \\
				1
				& $\beta_0=-1.75, \beta_1=0.05, \beta_2=0.5, \beta_3=0.75$
				& 0.321
				& 0.134 & 0.270 & 0.349 & 0.384
				& 0.135 & 0.224 & 0.257 & 0.305 \\
				& $\beta_0=-1.75, \beta_1=0.05, \beta_2=0.05, \beta_3=0.75$
				& 0.321
				& 0.132 & 0.220 & 0.312 & 0.350
				& 0.133 & 0.213 & 0.258 & 0.303 \\
				& $\beta_0=-1.75, \beta_1=0.05, \beta_2=0.05, \beta_3=0.05$
				& 0.320
				& 0.132 & 0.195 & 0.215 & 0.241
				& 0.132 & 0.227 & 0.252 & 0.292 \\
				\hline
				& $\beta_0=-1.75, \beta_1=0.5, \beta_2=1, \beta_3=2$
				& 0.650
				& 0.160 & 0.552 & 0.692 & 0.720
				& 0.165 & 0.112 & 0.156 & 0.290 \\
				2
				& $\beta_0=-1.75, \beta_1=0.05, \beta_2=0.5, \beta_3=0.75$
				& 0.643
				& 0.143 & 0.245 & 0.294 & 0.285
				& 0.143 & 0.100 & 0.118 & 0.142 \\
				& $\beta_0=-1.75, \beta_1=0.05, \beta_2=0.05, \beta_3=0.75$
				& 0.641
				& 0.138 & 0.179 & 0.263 & 0.235
				& 0.139 & 0.096 & 0.112 & 0.140 \\
				& $\beta_0=-1.75, \beta_1=0.05, \beta_2=0.05, \beta_3=0.05$
				& 0.638
				& 0.133 & 0.090 & 0.100 & 0.115
				& 0.133 & 0.100 & 0.110 & 0.126 \\
				\hline
			\end{tabular}
			\caption{Root Mean Squared Errors (RMSE) for the three factor loadings jointly considered from the IPW estimator; for the propensity score coefficient estimates from the proposed IPW estimator and from the Oracle IPW (based on true factor loadings), across alternative balancing scenarios, for Case 1 and Case 2. In all simulation cases, the panel dimensions are fixed at $N=500$ and $T=100$.}
			\label{Res_Sim_beta_loadings_oracle}
		\end{center}
	}
\end{table}

\subsection{Additional simulation designs}

To further assess the robustness of the proposed estimator, we extend the simulation study along three additional dimensions. First, we investigate the role of the time dimension $T$ and the cross-sectional size $N$. Second, we examine the consequences of misspecifying the number of latent factors in the propensity score model. Third, we consider simulation designs in which the idiosyncratic error term exhibits weak serial and cross-sectional dependence, thereby relaxing the baseline assumption of i.i.d. errors. In all the additional simulation designs considered in this section, the number of replications is fixed at 1000.

We begin by varying both the time dimension and the cross-sectional size. 
In addition to the baseline design with $T = 100$, we consider alternative values of the time dimension, $T = 50$ and $T = 200$. We also examine smaller cross-sectional sizes, $N = 50$ and $N = 150$, while keeping the baseline data-generating process unchanged, with the true number of factors equal to $r = 3$ and $\tau_{ATT} = 2$. For each configuration, we focus on the scenario that generates the highest degree of imbalance in the factor loadings, for both Case~1 and Case~2. By concentrating on the least favourable scenario, we assess the robustness of the estimator under severe imbalance conditions, thereby providing a conservative evaluation of its performance.
Tables \ref{Res_Sim_ATT_case2_B1_varying_N_T} and \ref{Res_Sim_beta_loadings_varying_N_T} show that, although the precision of the estimated loadings and propensity score coefficients improves as $T$ increases, since the estimation of the factor loadings benefits from a longer time series, this improvement has little impact on the RMSE of the ATT estimator, whose reduction instead depends on the sample size $N$, as the ATT is estimated through cross-sectional averaging. This is consistent with the asymptotic result in Theorem 1, according to which the additional variability induced by the estimation of the factor loadings and the propensity score parameters does not contribute to the first-order variance of the ATT estimator.

\begin{table}[!th]
	\small{
		\begin{center}
			\begin{tabular}
				{c|c|
					c @{\hspace{1\tabcolsep}} c |
					c @{\hspace{1\tabcolsep}} c |
					c @{\hspace{1\tabcolsep}} c}
				\hline
				& & \multicolumn{2}{c|}{$N=150$} 
				& \multicolumn{2}{c|}{$N=250$} 
				& \multicolumn{2}{c}{$N=500$}\\
				\hline
				\text{Case} & $T$ 
				& RMSE & Coverage 
				& RMSE & Coverage 
				& RMSE & Coverage \\
				& 
				& & of 95\% CI 
				& & of 95\% CI 
				& & of 95\% CI\\
				\hline
				\multirow{3}{*}{1} 
				& 50  & 0.513 & 0.933 & 0.390 & 0.924 & 0.284 & 0.913 \\
				& 100 & 0.498 & 0.934 & 0.384 & 0.932 & 0.274 & 0.937 \\
				& 200 & 0.521 & 0.930 & 0.394 & 0.924 & 0.271 & 0.948 \\
				\hline
				\multirow{3}{*}{2} 
				& 50  & 0.978 & 0.819 & 0.769 & 0.786 & 0.704 & 0.768 \\
				& 100 & 0.963 & 0.807 & 0.840 & 0.822 & 0.717 & 0.783 \\
				& 200 & 0.877 & 0.836 & 0.800 & 0.815 & 0.678 & 0.807 \\
				\hline
			\end{tabular}
			\caption{Root Mean Squared Errors (RMSE) and coverage of 95\% confidence intervals for the ATT from the proposed Inverse Propensity Score weighting method (IPW) with varying $N$ and $T$, under the scenario with parameters $\beta_0=-1.75$, $\beta_1=0.5$, $\beta_2=1$, and $\beta_3=2$.}
			\label{Res_Sim_ATT_case2_B1_varying_N_T}
		\end{center}
	}
\end{table}

\begin{table}[!th]
	\small{
		\begin{center}
			\begin{tabular}{c|c|c|c|cccc}
				\hline
				\text{Case} & $N$ & $T$ & Loadings & $\beta_0$ & $\beta_1$ & $\beta_2$ & $\beta_3$ \\
				\hline
				\multirow{9}{*}{1} 
				& 150 & 50  & 0.397 & 0.276 & 0.641 & 0.785 & 1.054 \\
				& 150 & 100 & 0.323 & 0.280 & 0.606 & 0.743 & 0.872 \\
				& 150 & 200 & 0.264 & 0.295 & 0.568 & 0.751 & 0.801 \\
				\cline{2-8}
				& 250 & 50  & 0.394 & 0.214 & 0.566 & 0.754 & 0.984 \\
				& 250 & 100 & 0.321 & 0.210 & 0.556 & 0.729 & 0.798 \\
				& 250 & 200 & 0.267 & 0.201 & 0.519 & 0.656 & 0.742 \\
				\cline{2-8}
				& 500 & 50  & 0.393 & 0.143 & 0.549 & 0.734 & 0.946 \\
				& 500 & 100 & 0.322 & 0.142 & 0.514 & 0.658 & 0.792 \\
				& 500 & 200 & 0.264 & 0.141 & 0.443 & 0.609 & 0.650 \\
				\hline \hline
				\multirow{9}{*}{2} 
				& 150 & 50  & 0.697 & 0.338 & 0.654 & 0.769 & 0.825 \\
				& 150 & 100 & 0.615 & 0.333 & 0.586 & 0.712 & 0.741 \\
				& 150 & 200 & 0.528 & 0.339 & 0.554 & 0.669 & 0.691 \\
				\cline{2-8}
				& 250 & 50  & 0.716 & 0.241 & 0.623 & 0.789 & 0.826 \\
				& 250 & 100 & 0.631 & 0.254 & 0.576 & 0.714 & 0.710 \\
				& 250 & 200 & 0.548 & 0.244 & 0.539 & 0.645 & 0.618 \\
				\cline{2-8}
				& 500 & 50  & 0.725 & 0.164 & 0.591 & 0.744 & 0.830 \\
				& 500 & 100 & 0.639 & 0.173 & 0.548 & 0.715 & 0.694 \\
				& 500 & 200 & 0.558 & 0.170 & 0.472 & 0.630 & 0.607 \\
				\hline
			\end{tabular}
			\caption{Root Mean Squared Errors (RMSE) for the three factor loadings jointly considered and for the propensity score coefficients estimates from the proposed IPW estimator with varying $N$ and $T$, under the scenario with parameters $\beta_0=-1.75$, $\beta_1=0.5$, $\beta_2=1$, and $\beta_3=2$.}
			\label{Res_Sim_beta_loadings_varying_N_T}
		\end{center}
	}
\end{table}

We next evaluate the sensitivity of the method to an incorrect specification of the number of factors entering the propensity score under the baseline design. The true data-generating process is driven by $r = 3$ latent factors. However, in the estimation of the propensity score, we alternatively impose $r = 2$ (underspecification) and $r = 4$ or $r = 5$ (overspecification), while leaving the outcome model unchanged. The third and fourth columns of Table {\ref{Res_Sim_ATT_case2_mispec_r1_combined}} report the RMSE and coverage of the ATT for the IPW estimator under these alternative specifications. When the propensity score is underspecified ($r = 2$), one relevant latent confounder is omitted from the treatment model, thereby violating the latent ignorability condition in the estimated specification. In this case, we observe an increase in bias and a deterioration in coverage, particularly in scenarios with stronger selection on the omitted factor (Case 2).
By contrast, when the propensity score is overspecified ($r = 4$ or $r = 5$), the inclusion of additional estimated factors that are not relevant for treatment assignment does not induce asymptotic bias, but may lead to a moderate loss of efficiency due to increased variability in the estimated propensity scores. The simulation results confirm that overspecification has limited impact on coverage, while underspecification is considerably more harmful.

As a third extension, we allow for weak dependence in the idiosyncratic errors. Indeed, the asymptotic normality result in Theorem 1 remains valid whether the idiosyncratic errors $\xi_{it}$ are i.i.d. or exhibit weak serial and cross-sectional dependence, provided that the latter satisfy the standard regularity conditions for approximate factor models \citep{bai2013principal}. These ensure $\sqrt{T_0}$-consistent estimation of the factor loadings and the applicability of the central limit theorem to cross-sectional averages; in essence, stationary short-memory processes and weak cross-sectional dependence are allowed. Thus, we build on the baseline design described above, but generate a stationary and serially dependent error according to the following specification: $\boldsymbol{\xi}_t = \phi_\xi \ \boldsymbol{\xi}_{t-1} + \boldsymbol{v}_t$, where $\boldsymbol{\xi}_t=(\xi_{1,t}, \cdots , \xi_{N,t})$ and where we set $\phi_\xi = 0.5$. Cross-sectional dependence is then introduced through a Toeplitz covariance matrix,

\begin{equation*}
	\boldsymbol{\Sigma}_\xi = \big(\tau^{|i-j|}\big)_{i,j=1}^N,
	\label{eq:Sigmae_toeplitz}
\end{equation*}

with $\tau = 0.5$. In particular $\boldsymbol{\Sigma}_v = (1-\phi_\xi^2)\boldsymbol{\Sigma}_\xi,$ and we draw $\boldsymbol{v}_t \stackrel{iid}{\sim} \mathcal{N}(0,\boldsymbol{\Sigma}_v)$ for $t=1,\ldots,T$. Tables 6 and 7 report the performance of the proposed IPW estimator and the widely used GSC and SDID\footnote{Both GSC and SDID allow for weak serial dependence in the idiosyncratic errors. However, their theoretical frameworks do not accommodate cross-sectional dependence in these errors, whereas the proposed IPW approach allows for both weak serial and cross-sectional dependence.} estimators according to the specified data-generating process, and under i.i.d. errors.
Table \ref{Res_Sim_ATT_allN_combined} compares the performance of the proposed IPW estimator with GSC and SDID under both i.i.d. and weakly dependent errors. The results reveal different patterns across methods and simulation designs. For GSC, moving from i.i.d. to weakly dependent errors systematically increases the RMSE in both Case 1 and Case 2, a pattern that persists when the number of factors is misspecified (Table \ref{Res_Sim_ATT_case2_mispec_r1_combined}). This deterioration may reflect the presence of cross-sectional dependence in the idiosyncratic errors, which can make it more difficult to separate the common factor structure from the remaining idiosyncratic component. By contrast, SDID displays lower RMSE under weak dependence than under i.i.d. errors throughout Case 1, while the differences are generally smaller in Case 2. A possible explanation is that the moderate serial dependence considered here increases the informativeness of pre-treatment outcomes for predicting post-treatment outcomes, which SDID can exploit through its unit and time weighting schemes. The moderate and rapidly decaying cross-sectional dependence considered here does not appear to offset this mechanism. The proposed IPW estimator is comparatively stable across the two error structures, with no systematic deterioration in RMSE under weak dependence. Overall, in Case 1 the RMSEs of the three methods are broadly comparable, particularly under weak dependence. In Case 2, instead, IPW exhibits larger RMSEs, while GSC yields the smallest RMSEs and SDID generally lies between the two. These results indicate that introducing moderate serial and cross-sectional dependence does not affect the three estimators uniformly and, in particular, does not necessarily lead to a deterioration in finite-sample performance. Although GSC and SDID may display lower RMSE and tighter confidence intervals than IPW, this apparent gain in precision must be interpreted with caution. Indeed, GSC and SDID do not account for the source of randomness arising from sampling from the super-population, and therefore it overstates estimator precision when the inferential target is the super-population rather than the observed sample.

\begin{table}[!th]
	\scriptsize{
		\begin{center}
			\setlength{\tabcolsep}{2pt}
			\renewcommand{\arraystretch}{1.05}
			\resizebox{\textwidth}{!}{%
				\begin{tabular}
					{c|c|c|
						cc|cc|
						cc|cc|
						cc|cc}
					\hline
					& & &
					\multicolumn{4}{c|}{IPW}
					& \multicolumn{4}{c|}{GSC}
					& \multicolumn{4}{c}{SDID} \\
					\hline
					\text{Case} & $N$ & $T$
					& \multicolumn{2}{c|}{i.i.d. errors}
					& \multicolumn{2}{c|}{weak dep. errors}
					
					& \multicolumn{2}{c|}{i.i.d. errors}
					& \multicolumn{2}{c|}{weak dep. errors}
					
					& \multicolumn{2}{c|}{i.i.d. errors}
					& \multicolumn{2}{c}{weak dep. errors} \\
					
					\cline{4-15}
					
					& &
					& RMSE & Cov.
					& RMSE & Cov.
					& RMSE & Cov.
					& RMSE & Cov.
					& RMSE & Cov.
					& RMSE & Cov. \\
					
					& &
					& & of 95\% CI
					& & of 95\% CI
					& & of 95\% CI
					& & of 95\% CI
					& & of 95\% CI
					& & of 95\% CI \\
					
					\hline
					
					\multirow{9}{*}{1}
					
					& 150 & 50
					& 0.513 & 0.933
					& 0.524 & 0.910
					& 0.423 & 0.918
					& 0.507 & 0.854
					& 0.507 & 0.915
					& 0.453 & 0.918 \\
					
					& 150 & 100
					& 0.498 & 0.934
					& 0.475 & 0.929
					& 0.409 & 0.915
					& 0.460 & 0.880
					& 0.484 & 0.937
					& 0.415 & 0.930 \\
					
					& 150 & 200
					& 0.521 & 0.930
					& 0.486 & 0.921
					& 0.417 & 0.915
					& 0.470 & 0.862
					& 0.493 & 0.914
					& 0.405 & 0.934 \\
					
					\cline{2-15}
					
					& 250 & 50
					& 0.390 & 0.924
					& 0.424 & 0.893
					& 0.329 & 0.946
					& 0.408 & 0.886
					& 0.397 & 0.915
					& 0.365 & 0.900 \\
					
					& 250 & 100
					& 0.384 & 0.932
					& 0.387 & 0.928
					& 0.315 & 0.943
					& 0.374 & 0.893
					& 0.381 & 0.915
					& 0.333 & 0.924 \\
					
					& 250 & 200
					& 0.394 & 0.924
					& 0.383 & 0.928
					& 0.327 & 0.931
					& 0.363 & 0.897
					& 0.383 & 0.899
					& 0.323 & 0.937 \\
					
					\cline{2-15}
					
					& 500 & 50
					& 0.296 & 0.902
					& 0.313 & 0.889
					& 0.230 & 0.940
					& 0.277 & 0.901
					& 0.310 & 0.872
					& 0.271 & 0.890 \\
					
					& 500 & 100
					& 0.274 & 0.937
					& 0.285 & 0.906
					& 0.229 & 0.948
					& 0.270 & 0.892
					& 0.281 & 0.898
					& 0.258 & 0.902 \\
					
					& 500 & 200
					& 0.272 & 0.934
					& 0.274 & 0.942
					& 0.226 & 0.944
					& 0.270 & 0.892
					& 0.265 & 0.918
					& 0.238 & 0.924 \\
					
					\hline
					\hline
					
					\multirow{9}{*}{2}
					
					& 150 & 50
					& 0.939 & 0.808
					& 0.951 & 0.810
					& 0.432 & 0.918
					& 0.527 & 0.829
					& 0.614 & 0.827
					& 0.594 & 0.833 \\
					
					& 150 & 100
					& 0.943 & 0.807
					& 0.896 & 0.851
					& 0.409 & 0.917
					& 0.487 & 0.858
					& 0.551 & 0.861
					& 0.523 & 0.861 \\
					
					& 150 & 200
					& 0.884 & 0.839
					& 0.889 & 0.811
					& 0.375 & 0.936
					& 0.494 & 0.830
					& 0.524 & 0.872
					& 0.499 & 0.847 \\
					
					\cline{2-15}
					
					& 250 & 50
					& 0.769 & 0.786
					& 0.854 & 0.789
					& 0.321 & 0.940
					& 0.416 & 0.862
					& 0.516 & 0.829
					& 0.514 & 0.794 \\
					
					& 250 & 100
					& 0.840 & 0.822
					& 0.816 & 0.824
					& 0.321 & 0.930
					& 0.392 & 0.864
					& 0.458 & 0.849
					& 0.435 & 0.850 \\
					
					& 250 & 200
					& 0.800 & 0.815
					& 0.802 & 0.811
					& 0.308 & 0.929
					& 0.366 & 0.878
					& 0.419 & 0.861
					& 0.398 & 0.857 \\
					
					\cline{2-15}
					
					& 500 & 50
					& 0.731 & 0.759
					& 0.696 & 0.765
					& 0.262 & 0.893
					& 0.295 & 0.853
					& 0.445 & 0.755
					& 0.427 & 0.746 \\
					
					& 500 & 100
					& 0.706 & 0.763
					& 0.688 & 0.796
					& 0.238 & 0.912
					& 0.289 & 0.827
					& 0.354 & 0.835
					& 0.343 & 0.827 \\
					
					& 500 & 200
					& 0.671 & 0.818
					& 0.702 & 0.806
					& 0.221 & 0.931
					& 0.276 & 0.852
					& 0.314 & 0.842
					& 0.296 & 0.856 \\
					
					\hline
					
				\end{tabular}%
			}
			
			\caption{Root Mean Squared Errors (RMSE) and coverage of 95\% confidence intervals for the ATT from the proposed IPW estimator, the Generalized Synthetic Control method (GSC), and Synthetic Difference-in-Differences (SDID) with varying $N$ and $T$, under the scenario with parameters $\beta_0=-1.75$, $\beta_1=0.5$, $\beta_2=1$, and $\beta_3=2$. The results are presented for both i.i.d. and weakly dependent errors.}
			\label{Res_Sim_ATT_allN_combined}
		\end{center}
	}
\end{table}

\begin{table}[!th]
	\scriptsize
	\begin{center}
		\resizebox{\textwidth}{!}{%
			\begin{tabular}{c|c|cc|cc|cc|cc}
				\hline
				\multicolumn{2}{c|}{} 
				& \multicolumn{4}{c|}{IPW} 
				& \multicolumn{4}{c}{GSC} \\
				\hline
				\text{Case} & $r_1$
				& \multicolumn{2}{c|}{i.i.d. errors}
				& \multicolumn{2}{c|}{weak dep. errors}
				& \multicolumn{2}{c|}{i.i.d. errors}
				& \multicolumn{2}{c}{weak dep. errors} \\
				\cline{3-10}
				& 
				& RMSE & Coverage
				& RMSE & Coverage
				& RMSE & Coverage
				& RMSE & Coverage \\
				& 
				& & of 95\% CI
				& & of 95\% CI
				& & of 95\% CI
				& & of 95\% CI \\
				\hline
				1 & 2 
				& 0.439 & 0.751 
				& 0.436 & 0.747 
				& 0.356 & 0.819 
				& 0.385 & 0.777 \\
				& 3 
				& 0.274 & 0.937 
				& 0.285 & 0.906 
				& 0.229 & 0.942 
				& 0.270 & 0.892 \\
				& 4 
				& 0.271 & 0.931 
				& 0.282 & 0.921 
				& 0.220 & 0.957 
				& 0.260 & 0.909 \\
				& 5 
				& 0.282 & 0.930 
				& 0.284 & 0.914 
				& 0.228 & 0.940 
				& 0.263 & 0.908 \\
				\hline
				2 & 2 
				& 1.451 & 0.420 
				& 1.429 & 0.419 
				& 0.835 & 0.565 
				& 0.861 & 0.555 \\
				& 3 
				& 0.717 & 0.783 
				& 0.712 & 0.762 
				& 0.239 & 0.922 
				& 0.298 & 0.838 \\
				& 4 
				& 0.680 & 0.794 
				& 0.665 & 0.768 
				& 0.227 & 0.933 
				& 0.284 & 0.861 \\
				& 5 
				& 0.716 & 0.796 
				& 0.659 & 0.772 
				& 0.231 & 0.933 
				& 0.277 & 0.864 \\
				\hline
			\end{tabular}%
		}
		\caption{Root Mean Squared Errors (RMSE) and coverage of 95\% confidence intervals for the ATT from the proposed IPW estimator and the Generalized Synthetic Control method (GSC) varying the number of factors used in estimation ($r_1$), under the scenario with parameters $\beta_0=-1.75$, $\beta_1=0.5$, $\beta_2=1$, and $\beta_3=2$. In all simulation cases, the panel dimensions are fixed at $N=500$ and $T=100$. The results are presented for both weakly dependent and i.i.d. errors.}
		\label{Res_Sim_ATT_case2_mispec_r1_combined}
	\end{center}
\end{table}

\section{The short-run effect of Paris Agreement on European Stock Returns}
Our methodology is applied to estimate the causal effect of the Paris Agreement on European stock returns in the first quarter of 2016. Specifically, we evaluate market reactions of firms that demonstrate a strong commitment to sustainability after the policy implementation.

In this context, a proper definition of the treatment variable is crucial. In particular, two considerations are relevant. First, the issuance of green bonds serves as a credible signal of a persistent strategic positioning toward environmentally sustainable financing in the post–Paris Agreement period \citep{flammer2021corporate}. Second, the launch of a green bond programme is a multi-stage process that requires early decisions on eligible projects and investment pipelines, adjustments in governance and reporting systems, and coordination with underwriting banks and external reviewers to structure, certify, and place the security. Consequently, although green bonds are often issued only at a later stage, credible information about this strategic realignment may become available to market participants earlier and can therefore be incorporated into stock prices in the immediate post-Agreement period. Accordingly, from our ex-post perspective, we classify as treated units the European companies that issued at least one green bond during the period 2016–2019.

In particular, our definition is based on two considerations. First, the issuance of green bonds serves as a credible signal of a persistent strategic positioning toward environmentally sustainable financing in the post–Paris Agreement period.

To this end, two datasets are obtained from Bloomberg. We download the quarterly adjusted closing prices for S$\&$P350 Europe\footnote{S$\&$P350 Europe consists of 350 leading blue-chip companies drawn from 16 developed European markets, mirroring the sector and country weights of a broader European financial universe.} constituents from 30/06/1998 to 31/03/2016 and collect a dataset of EU corporate green bonds by retrieving all securities labelled 'green' in Bloomberg's fixed income database\footnote{Precisely, green bonds are defined as those for which the 'Green Bond Indicator' field in Bloomberg's fixed income database is set to 'Yes'.}. 
After merging the two datasets and excluding companies that issued a green bond before 2016, we are left with a final sample consisting of 29 treated units out of 224 companies.

For our analysis, we compute the standardized log returns before treatment and transform them to ensure each series has a mean zero. We then estimate via PCA an approximate factor model driven by three factors following $IC_1$ and $IC_2$ criteria of \cite{bai2002determining}. Based on the estimated loadings, $\boldsymbol{\hat{\lambda}}_i$, we compute the propensity score through a logistic regression. The proposed method leads to the following fitted model $\Pr(Z_{i,T} = 1|\hat{\lambda}_i) = -3.02 - 8.06 \, \hat{\lambda}_{i1} -8.27 \,\hat{\lambda}_{i2} - 17.03 \, \hat{\lambda}_{i3}$, with the adjusted standard error of the parameter estimators equal to $(0.76, 6.48, 4.50, 10.78)^{'}$.
To validate the inferential procedure, we assess the validity of the overlap assumption by analyzing the distribution of estimated propensity scores across the treatment and control groups. Figure \ref{fig:Figure Prop_Score} displays the histograms of the scores allowing for a visual inspection of the common support. The results indicate a satisfactory degree of overlap between the two groups.  Furthermore, the reliability of the estimated propensity scores hinges on the balancing of the weighted distributions of loadings across treatment and control groups that is a consequence of latent ignorability (Assumption \ref{Assumption 4}). To this end, we compute the absolute standardized difference (ASD). Table \ref{tab:asd_table} illustrates that the weighting significantly enhances the overall balance of the estimated loadings, with the maximum ASD value remaining below 0.60, much smaller than the critical threshold 1.96.

The estimated average treatment effect on the treated (ATT) is –0.070 with a standard error of 0.027 ($p$-value: 0.009). This suggests that the Paris Agreement had a statistically significant negative effect on green European stock returns in the first quarter of 2016. This conclusion is confirmed by alternative estimators commonly used in policy evaluation, including difference-in-differences (DiD), generalized synthetic control (GSC; \citealp[]{xu2017generalized}), and synthetic difference-in-differences (SynthDiD; \citealp[]{arkhangelsky2021synthetic}). As reported in Table \ref{tab:diff_methods}, the estimated effects have the same sign across methods, with DiD yielding an almost identical estimate and SynthDiD and GSC producing a smaller absolute effects. The differences in magnitude across estimators can be interpreted as a consequence of the relative importance assigned to selection into the treatment versus the fit of pre-treatment outcome dynamics. Our IPW estimator is a design-based causal method based on unconfoundedness on factor loadings, whereas SynthDiD and GSC prioritize the ability to reproduce the pre-treatment time pattern of the outcome. At the same time, the fact that the standard DiD estimator produces a result that is nearly identical to our IPW estimate suggests that, in this application, deviations from the parallel trends assumption are likely limited. The following Section 5.1.2 provides a further comparison of the results produced by these estimators as a sensitivity analysis.

The negative short-run effect reflects investors’ reassessment of the trade-off between financial performance and environmental commitment, triggered by increased awareness of climate-related risks following the Agreement. Green firms tend to trade at a premium compared to brown firms, reflecting their perceived lower risk and stronger alignment with sustainability objectives. Consequently, investors appear willing to accept lower expected returns in exchange for exposure to pro-environmental stocks. This behavior reflects the role of the Paris Agreement as a strong and credible policy signal that has reshaped investor expectations about climate change. By setting clear and long-term decarbonization objectives, the Agreement has prompted financial markets to reallocate capital toward climate-aligned sectors, supporting investment in clean technologies and sustainable infrastructure. This shift has also triggered a broader structural transformation, as market participants have revised traditional valuation models to incorporate climate-related risks and assess emerging investment opportunities. As a result, financial assets have been repriced to reflect the expected impact of the low-carbon transition on future cash flows, operating costs, and regulatory exposure.

\begin{figure}[H]
	\centering  \includegraphics[width=0.5\linewidth, height=6.5cm]{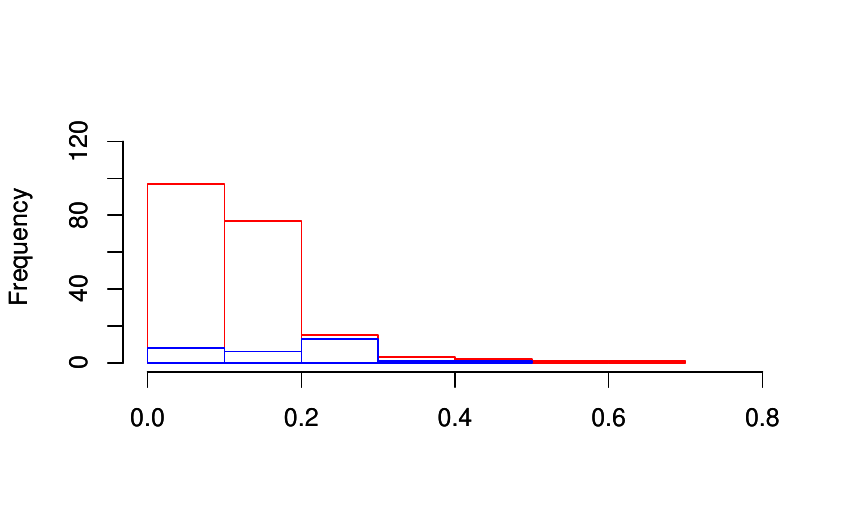}
	\caption{Histograms of the estimated propensity score for the treated (blue) and the untreated (red).}
	\label{fig:Figure Prop_Score}
\end{figure}

\vspace{0.3em}
\begin{table}[!th]
	\centering
	\small{
		\begin{tabular}
			{cc @{\hspace{1\tabcolsep}} c}
			\hline
			{ } & {Unweighted} & {Weighted} \\
			\hline
			${\hat{\lambda}}_1$ & 1.80 & 0.39 \\
			${\hat{\lambda}}_2$ & 2.48 & 0.52 \\
			${\hat{\lambda}}_3$ & 2.85 & 0.60 \\
			\hline
		\end{tabular}
		\medskip
		\caption{Absolute Standardized Difference, ASD, of estimated factor loadings in the original, unweighted, and weighted sample.}
		\label{tab:asd_table}
	}
\end{table}

\vspace{0.3em}

\begin{table}[!th]
	\centering
	\small{
		{
			\begin{tabular}
				{c|c @{\hspace{1\tabcolsep}} |c|c}
				\hline
				Estimator & ATT & SE & $p$-value \\
				\hline
				IPW   & -0.071 & 0.027 & 0.009 \\
				DiD & -0.072 & 0.020 &  0.0003 \\
				SynthDiD  & -0.038 & 0.018 & 0.035 \\
				GSC & -0.047 & 0.023 & 0.042 \\
				
				\hline
	\end{tabular}}}
	\medskip 
	\caption{Average treatment effect on the treated (ATT), standard errors (SE), and $p$-values computed using the proposed inverse propensity score weighting method (IPW), difference-in-differences (DiD), generalized synthetic control (GSC; \citealp[]{xu2017generalized}), and synthetic difference-in-differences (SynthDiD; \citealp[]{arkhangelsky2021synthetic}).}
	\label{tab:diff_methods}
\end{table}

\subsection{Robustness checks}

\subsubsection{Augmented propensity score with observable pre-treatment covariates}

As a first robustness check, we consider an augmented specification of the propensity score that includes a set of pre-treatment observable firm characteristics in addition to the estimated factor loadings. While our identification strategy relies on latent ignorability (Assumption \ref{Assumption 4}), conditioning on observable covariates provides a complementary assessment of the sensitivity of the results to observable sources of heterogeneity. Preliminary balance diagnostics based on observable firm characteristics indicate that, when using a propensity score specified solely as a function of the estimated factor loadings, several covariates remain unbalanced, with Absolute Standardized Differences, ASD, exceeding conventional thresholds even after weighting (Table \ref{tab:characteristics} in Appendix \ref{app: sec13_bis}). This evidence suggests that the latent factor structure, while effective in capturing relevant sources of unobserved heterogeneity, does not fully account for all observable differences between treated and control firms\footnote{See Appendix \ref{app: sec13_bis} for additional descriptive information on the country and sector composition of treated and control firms.}. Accordingly, we augment the propensity score specification with additional pre-treatment covariates\footnote{The augmented specification is estimated on the complete-case sample, as some observed covariates are missing for a small number of firms. The resulting sample is only slightly smaller than in the baseline specification and consists of 218 firms, including 28 treated units.}. However, to preserve the overlap condition, we restrict the number of included covariates. Given the relatively small number of treated units in our sample, indeed, including a large set of predictors would lead to highly concentrated propensity scores and extreme weights. Therefore, we extend the propensity score model by including a selected set of four firm-level characteristics\footnote{Alternative selected parsimonious specifications based on pre-treatment confounders were also explored. The specification adopted here emerged as the most satisfactory in terms of covariate balance, overlap, propensity score stability, retention of treated units, and economic interpretability.}: size, financial leverage, environmental performance, and profitability, proxied respectively by the logarithm of revenue, the total debt-to-assets ratio, the ESG score, and return on equity (ROE)\footnote{Revenue, ROE, and the total debt-to-assets ratio are observed on September 30, 2015. The ESG score, which is available only at annual frequency, refers to the 2015 reporting year and is reported on December 31, 2015.}. We evaluate the balancing properties of the augmented propensity score by computing the ASD for both the estimated factor loadings and the included observed covariates: Table \ref{tab:asd_with_covariates} shows that all ASDs for the weighted sample fall below the threshold value of 1.96. Under the augmented specification, the estimated ATT is $-0.072$, with a standard error of $0.037$ ($p$-value: $0.052$). The estimated effect is therefore consistent with the previous results obtained under the simpler specification based solely on factor loadings, both in magnitude and precision. This finding provides indirect support for this approach, suggesting that the latent factor loadings serve as a low-dimensional representation of unobserved heterogeneity and effectively capture the main sources of confounding.

\begin{table}[htbp]
	\centering
	\begin{tabular}{lcc}
		\hline
		& Unweighted & Weighted \\
		\hline
		\multicolumn{1}{c}{${\hat{\lambda}}_1$} & 1.580 & 0.223 \\
		\multicolumn{1}{c}{${\hat{\lambda}}_2$} & 2.475 & 0.040  \\
		\multicolumn{1}{c}{${\hat{\lambda}}_3$} & 2.532 & 0.434 \\
		\multicolumn{1}{c}{Log(Revenue)} & 6.242 & 0.067 \\
		\multicolumn{1}{c}{ROE} & 3.101 & 0.199 \\
		\multicolumn{1}{c}{$\frac{\text{Total Debt}}{\text{Total Assets}}$} & 3.302 & 0.057 \\
		\multicolumn{1}{c}{ESG score} & 2.117 & 0.939 \\
		\hline
	\end{tabular}
	\caption{Absolute Standardized Difference, ASD, of estimated factor loadings and selected pre-treatment covariates in the unweighted (original) and weighted sample.}
	\label{tab:asd_with_covariates}
\end{table}

\subsubsection{Falsification tests and diagnostics}

Our analysis may potentially capture not only the causal effect of the Paris Agreement but also the effects of other policies and/or dynamics already in place prior to its adoption. To investigate this possibility, we perform falsification tests using negative controls \citep{rosenbaum2002}. Each test involves conducting the analysis during the pre-treatment period, using a fictitious policy implementation date. Specifically, we run three falsification tests with different fictitious policy dates: 01/04/2011, 01/04/2013, and 01/04/2015. Since the Paris Agreement had not yet been enacted, we expect no effects from these ‘false’ policies or any other contemporaneous dynamics. The estimated effects\footnote{For each falsification test, we checked the overlap assumption as well as the balance of the weighted distribution of the estimated loadings. The corresponding propensity score histograms and tables reporting the absolute standardized differences (ASD) of the loadings are provided in Appendix \ref{app: sec13}}  for the pre-Paris Agreement sample are reported in Table \ref{tab:fals_test} for both the specificiation with only factor loadings and the augmented one adopted in the previous subsection, which combines the factor loadings with the selected set of four firm-level covariates (size, financial leverage, environmental performance, and profitability). The estimates are close to zero and not statistically significant, with very high $p$-values. These results support the conclusion that the effects identified in the main analysis are primarily attributable to the Paris Agreement, rather than to other confounding dynamics or policies.

\begin{table}[!th]
	\centering
	\small{
		\begin{tabular}{c|ccc}
			\hline
			\multicolumn{4}{c}{\textbf{Specification with factor loadings only}} \\
			\hline
			\multicolumn{1}{l}{Fictitious implementation date:} & {01/04/2011} & {01/04/2013} & {01/04/2015} \\
			\hline
			ATT & 0.023 & 0.013 & -0.017 \\
			p-value & 0.710 & 0.499 & 0.327 \\
			\hline \hline
			\multicolumn{4}{c}{\textbf{Specification with factor loadings and covariates}} \\
			\hline
			\multicolumn{1}{l}{Fictitious implementation date:} & {01/04/2011} & {01/04/2013} & {01/04/2015} \\
			\hline
			ATT & -0.005 & 0.017 & 0.006 \\
			p-value & 0.896 & 0.416 & 0.766 \\
			\hline
		\end{tabular}
	}
	\medskip
	\caption{Falsification tests based on different fictitious implementation dates of the treatment for the specificiation with only factor loadings and the augmented one with factor loadings and the selected firm-level pre-treatment covariates.}
	\label{tab:fals_test}
\end{table}

In addition to the negative control tests based on fictitious policy dates, we also assess the plausibility of the identifying assumptions underlying both the proposed IPW method and the alternative methods considered in the previous section (listed in Table \ref{tab:diff_methods}). To this end, we use the following longitudinal pre-treatment diagnostics: estimated placebo effects for our proposed IPW estimator, pre-treatment fit for SDID and GSC, and a pre-trends analysis for DiD. It is worth noting that, for IPW, this longitudinal placebo analysis differs from the previous falsification analysis based on negative controls. The latter involves a full re-estimation of the model, including the propensity score, under alternative fictitious policy dates, whereas the analysis below relies on the propensity score estimated in the main specification. Figure 4a shows the trajectories of these pre-treatment diagnostics up to time $T-1$, together with the estimated ATT at the single post-treatment time $T$. To facilitate visual inspection, rather than reporting the diagnostic values at each individual pre-treatment time point, we display continuous smooth curves obtained using LOESS (Locally Estimated Scatterplot Smoothing). It should also be noted that the baseline IPW and GSC specifications differ in the selected number of factors. While IC1 and IC2 select $r=3$ for IPW, cross-validation selects $r=5$ for GSC. To assess the relevance of this difference, we re-estimated IPW with $r=5$, as selected by IC3, and GSC with $r=3$. The resulting ATT estimates are approximately -0.04 and -0.07, respectively, suggesting that the difference between the baseline IPW and GSC estimates is substantially associated with the selected factor dimension rather than with the estimation method per se. Figure~\ref{fig:pre_treat_diagnostics}(b) therefore reports the corresponding pre-treatment diagnostics, where, unlike in Figure~\ref{fig:pre_treat_diagnostics}(a), the dashed lines represent the pre-treatment diagnostics for IPW and GSC with $r=5$ and $r=3$, respectively.\\ Figures~\ref{fig:pre_treat_diagnostics}(a) and \ref{fig:pre_treat_diagnostics}(b) provide broadly reassuring evidence from the pre-treatment diagnostics. For IPW, the estimated placebo effects remain close to zero throughout the pre-treatment period. An additional analysis, reported in Appendix \ref{app: sec13_ter}, shows that the large majority of the estimated placebo ATTs are statistically indistinguishable from zero for IPW also when the propensity-score specification is augmented with the four selected covariates considered in the previous Section. SDID and GSC display satisfactory pre-treatment fit, while for DiD there is no evidence of systematic pre-treatment trends. Taken together, these diagnostics do not indicate substantial pre-treatment discrepancies for any of the methods considered, thus supporting, in the current application, the plausibility of a low-rank factor structure with moderate serial and cross-sectional dependence in the errors and no marked departure from the parallel trends assumption. All four methods yield statistically significant negative estimates of the effect of the Paris Agreement, with point estimates ranging approximately from $-4 \%$ to $-7 \%$. The similarity of the estimates should not, however, obscure an important difference in their interpretation. In our framework, the IPW estimator identifies a super-population ATT, whereas the alternative estimators considered here identify finite-sample effects, namely the effects for the units represented in the observed sample. As discussed in the Introduction, this super-population perspective provides a different type of policy-relevant information, since the estimand is defined with respect to a population from which the observed firms are viewed as realizations rather than being restricted to the particular sample at hand.

\begin{figure}[!htbp]
	\centering
	\medskip
	\includegraphics[width=0.85\textwidth]{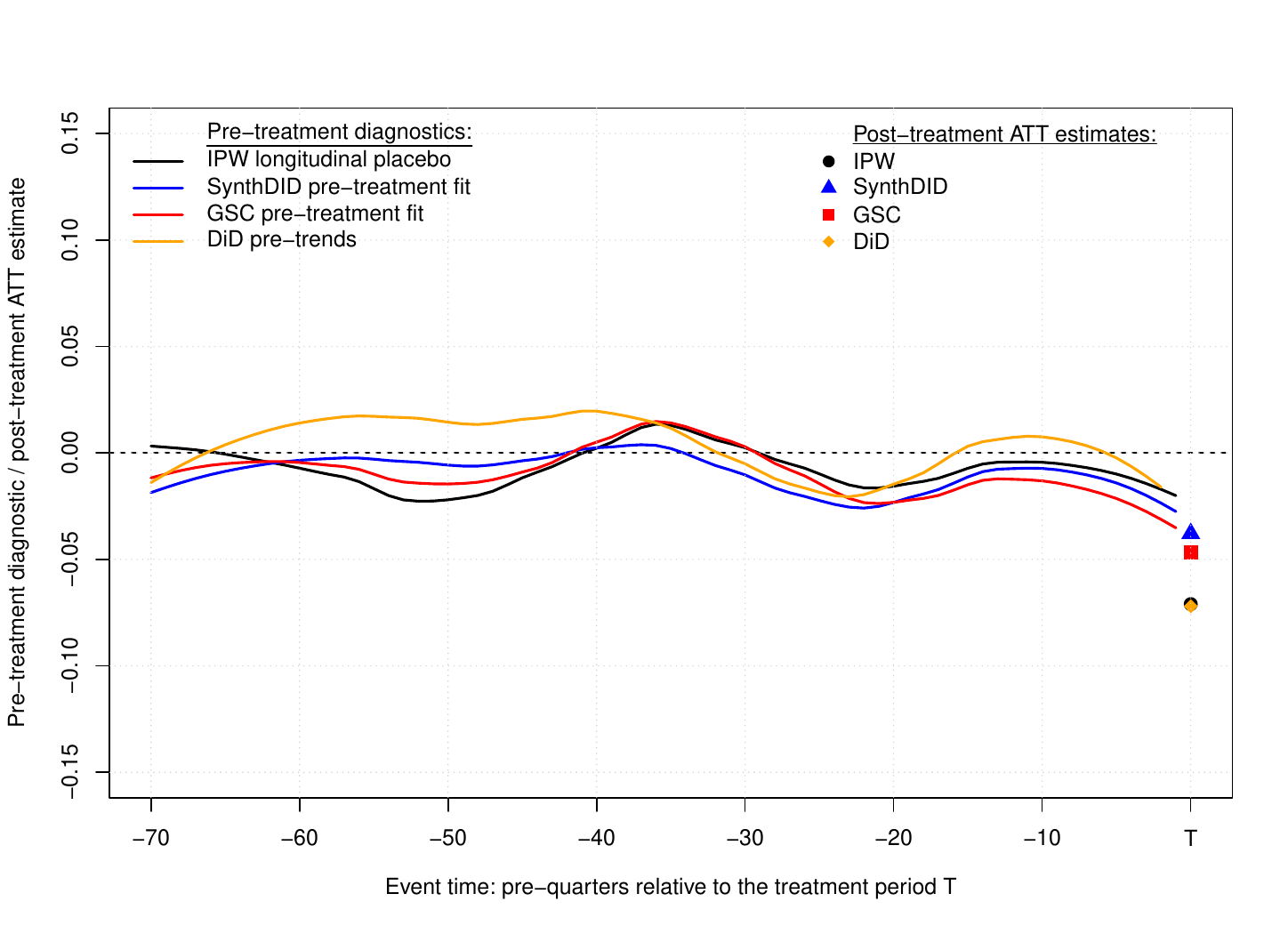}
	\smallskip
	\text{(a)}
	\vspace{0.6cm}
	\includegraphics[width=0.85\textwidth]{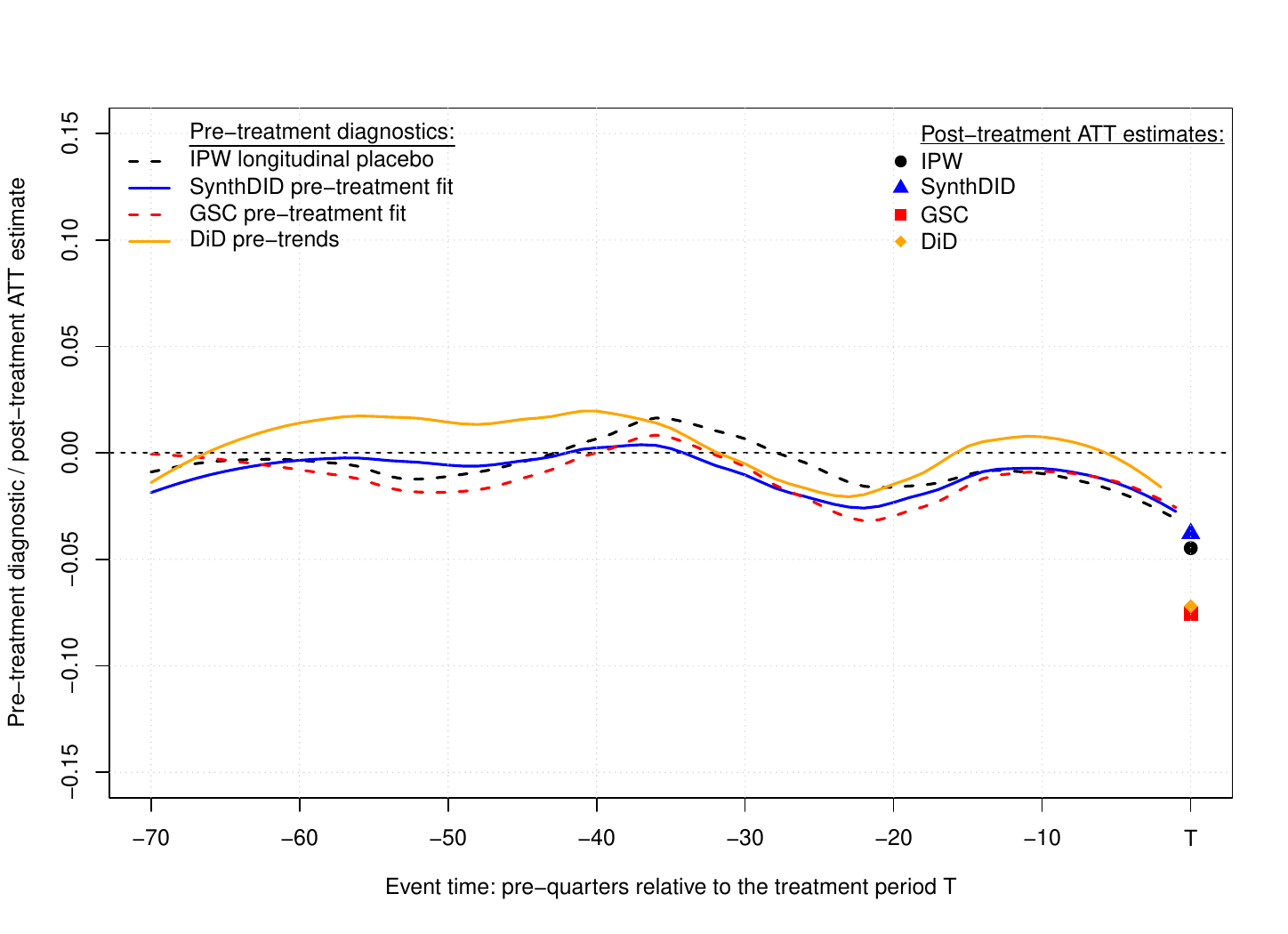}
	\smallskip
	\text{(b)}
	\caption{LOESS curves of the pre-treatment diagnostics and post-treatment ATT estimates for IPW, DiD, SDID, and GSC. (a) Baseline  pecifications, with $r=3$ for IPW, as selected by IC1 and IC2, and $r=5$ for GSC, as selected by cross-validation. (b) Alternative factor specifications, with $r=5$ for IPW and $r=3$ for GSC.}
	\label{fig:pre_treat_diagnostics}
\end{figure}

\section{Conclusion}

We investigate the short-run impact of the Paris Agreement on European equity markets, focusing on firms that signalled environmental commitment through green bond issuance. To this end, we propose a novel causal inference method for panel data based on inverse propensity score weighting where the propensity score is function of latent factor loadings. The approach enables inference under latent ignorability within a causal super-population mode of inference, without requiring assumptions such as parallel trends or homogeneous treatment effects, while ensuring a higher degree of external validity. In particular, the large-sample properties of the estimator are derived under M-estimation theory 
yielding a closed-form expression for the asymptotic variance of ATT. This allows for causal inference without the need for computationally intensive methods. 


We apply this methodology to panel data on the constituents of the S$\&$P 350 Europe index, defining as treated those firms that issued at least one green bond between 2016 and 2019. Latent factor loadings are estimated via principal component analysis on pre-treatment returns and used to compute propensity scores through a logistic regression. We find a statistically significant negative impact of the Paris Agreement on the returns of treated firms in the first quarter of 2016 indicating that the issuing firms experienced lower returns following the policy. This behaviour reflects investors’ reassessment of the trade-off between financial performance and environmental commitment. Green firms tend to trade at a premium relative to brown firms, as market participants appear willing to accept lower expected financial returns in exchange for exposure to climate policy-aligned assets, reflecting their perceived lower risk. This underscores the role of the Paris Agreement as a strong and credible policy signal that has reshaped investor expectations about climate change, fostering capital reallocation toward low-carbon sectors.


Our approach contributes to the growing literature on causal inference with latent structures and offers advantages in terms of causal interpretability, policy relevance, and computational efficiency, while relaxing the traditional restrictive assumptions of causal panel methods.

A limitation of the proposed method lies in the limited economic interpretability of the estimated factor loadings. Although they capture unit-specific characteristics, their structure tends to be dense with many large loadings across factors, making interpretation difficult. This is left for future research where our goal is to enable an economic interpretation through factor rotation techniques. 
We also acknowledge additional limitations of our work, for which deriving the asymptotic distribution would have been challenging to address within a single paper. We therefore leave these issues as directions for future research. From an applied perspective, an important extension concerns adapting the proposed framework to settings with staggered or sequential treatment adoption. From a methodological perspective, a further avenue for future research is the extension of our approach to a doubly robust framework.

\begin{appendices}
	
	\section{Appendix}\label{sec11}
	
	\subsection{ Proof of Theorem~{\upshape\ref{thm1}}}\label{sec12}
	
	To prove the theorem we need the following lemma:
	
	\begin{lemma}\label{lemma1} Assume that the propensity score, $e_i$, follows a logistic model where the covariates are the loadings, $\boldsymbol{\hat{\lambda}}_i$, obtained via PCA, and coefficients $\boldsymbol{\beta}$ are estimated using maximum likelihood. Under standard assumptions for approximate factor models (\citealp[]{bai2013principal}), for $N, T_{0} \rightarrow \infty$ with $\sqrt{T_{0}}/N \rightarrow 0$, we have
		
		$$
		\begin{aligned}
			\sqrt{N}(\hat{\boldsymbol{\beta}}-\boldsymbol{\beta}) \xrightarrow{d} \mathcal{N}\left(0,\mathbb{V}^{\boldsymbol{\beta}}\right),
		\end{aligned}
		$$
		
		where 
		
		$$
		\begin{aligned}
			\mathbb{V}^{\boldsymbol{\beta}}  =\boldsymbol{E}_{\beta \beta}^{-1}\,\,\mathbb{E}\left[ \boldsymbol{S}_i\left(\boldsymbol{\beta},\boldsymbol{{\lambda}_i}\right)\boldsymbol{S}_i\left(\boldsymbol{\beta},\boldsymbol{{\lambda}_i}\right)^{\mathcal{T}}\right]\,\boldsymbol{E}_{\beta \beta}^{-1},
		\end{aligned}
		$$
		with $\boldsymbol{S}_i\left(\boldsymbol{\beta},\boldsymbol{{\lambda}_i}\right)$ denoting  the score function for the propensity score and $\boldsymbol{E}_{\beta \beta} = {-}\,\mathbb{E}[\frac{\partial \boldsymbol{S}_{i}\left(\boldsymbol{\beta},\boldsymbol{\lambda}_i \right)}{\partial \boldsymbol{\beta}^{\mathcal{T}}}]$  the information matrix for the propensity score.
	\end{lemma}
	
	\begin{proof}[Proof of Lemma \ref{lemma1}.]
		
		For a logistic regression of the form $e(\boldsymbol{\hat{\lambda}}_i) = 1/1+\exp (-{\boldsymbol{\hat{\lambda}}_i'\boldsymbol{\beta})}$ estimated by maximum likelihood, the mean value expansion of the first-order condition, $\sum_{i=1}^{N} \boldsymbol{S}_i(\boldsymbol{\hat{\beta}},\boldsymbol{{\hat{\lambda}}}_i)=0$, around the true values ($\boldsymbol{\beta, \lambda}$) is 
		
		\begin{equation}
			\begin{aligned}
				0 = &  \sum_{i=1}^{N} \boldsymbol{S}_i(\boldsymbol{\hat{\beta}},\boldsymbol{{\hat{\lambda}}}_i) \\
				= & \sum_{i=1}^{N}\boldsymbol{S}_i(\boldsymbol{\beta},{\boldsymbol{\lambda}_i}) +  \sum_{i=1}^{N} \frac{\partial \boldsymbol{S}_{i}\left(\boldsymbol{\beta},\boldsymbol{\lambda}_i \right)}{\partial \boldsymbol{\beta}^{\mathcal{T}}}|_{\boldsymbol{\beta}=\tilde{\boldsymbol{\beta}}} (\boldsymbol{\hat{\beta}} - \boldsymbol{{\beta}}) \\
				& + \sum_{i=1}^{N} \frac{\partial \boldsymbol{S}_{i}\left(\boldsymbol{\beta},\boldsymbol{\lambda}_i \right)}{\partial \boldsymbol{\lambda}^{{\mathcal{T}}}}|_{\boldsymbol{\lambda}=\tilde{\boldsymbol{\lambda}}} (\boldsymbol{\hat{{\lambda}}}_i-\boldsymbol{\lambda}_i).
				\label{eq:mv-expansion}
			\end{aligned} 
		\end{equation} 
		
		where $\boldsymbol{S}_i(\boldsymbol{\hat{\beta}},\boldsymbol{\hat{\lambda}}_i)$ represents the score function, while $\boldsymbol{\tilde{\beta}}$ and $\boldsymbol{\tilde{\lambda}}_i$ are the mean values. 
		{Dividing both sides of \eqref{eq:mv-expansion} by $\sqrt{N}$, we obtain
			
			\begin{equation} 
				\label{eq:mv-expansion-normalized}
				\begin{aligned}
					\boldsymbol{0}
					= &\frac{1}{\sqrt{N}}\sum_{i=1}^{N}\boldsymbol{S}_i(\boldsymbol{\beta},\boldsymbol{\lambda}_i)
					+  \frac{1}{N}\sum_{i=1}^{N} 
					\left.\frac{\partial \boldsymbol{S}_{i}\left(\boldsymbol{\beta},\boldsymbol{\lambda}_i \right)}
					{\partial \boldsymbol{\beta}^{\mathcal{T}}}\right|_{\boldsymbol{\beta}=\tilde{\boldsymbol{\beta}}}
					\sqrt{N}\,(\widehat{\boldsymbol{\beta}} - \boldsymbol{\beta}) \\
					&\quad + \frac{1}{\sqrt{N}} \sum_{i=1}^{N}  
					\left.\frac{\partial \boldsymbol{S}_{i}\left(\boldsymbol{\beta},\boldsymbol{\lambda}_i \right)}
					{\partial \boldsymbol{\lambda}^{\mathcal{T}}}\right|_{\boldsymbol{\lambda}=\tilde{\boldsymbol{\lambda}}_i}
					(\widehat{\boldsymbol{\lambda}}_i-\boldsymbol{\lambda}_i).
				\end{aligned} 
			\end{equation} 
			
			Solving \eqref{eq:mv-expansion-normalized} for $\sqrt{N}(\widehat{\boldsymbol{\beta}}-\boldsymbol{\beta})$ yields
			\begin{equation} 
				\begin{aligned}
					\sqrt{N}(\widehat{\boldsymbol{\beta}}-\boldsymbol{\beta})
					= {}&
					\widehat{\boldsymbol{E}}_{\beta \beta}^{-1}\,
					\frac{1}{\sqrt{N}} \, \Bigg( \sum_{i=1}^{N}\boldsymbol{S}_i(\boldsymbol{\beta},\boldsymbol{\lambda}_i)
					\\
					&\quad
					+ \sum_{i=1}^{N}  
					\left.\frac{\partial \boldsymbol{S}_{i}\left(\boldsymbol{\beta},\boldsymbol{\lambda}_i \right)}
					{\partial \boldsymbol{\lambda}^{\mathcal{T}}}\right|_{\boldsymbol{\lambda}=\tilde{\boldsymbol{\lambda}}_i}
					(\widehat{\boldsymbol{\lambda}}_i-\boldsymbol{\lambda}_i)
					\Bigg).
					\label{equation:beta_exp}
				\end{aligned} 
			\end{equation}
			
			where $\widehat{\boldsymbol{E}}_{\beta\beta}
			= {-}
			\frac{1}{N}\sum_{i=1}^{N} 
			\left.\frac{\partial \boldsymbol{S}_{i}\left(\boldsymbol{\beta},\boldsymbol{\lambda}_i \right)}
			{\partial \boldsymbol{\beta}^{\mathcal{T}}}\right|_{\boldsymbol{\beta}=\tilde{\boldsymbol{\beta}}}.$
			
			In Equation (A.5) and Theorem 1, under the identifying restrictions and standard assumptions,
			\cite{bai2013principal} show that if $\sqrt{T_{0}} /N \rightarrow 0$ and
			$H^{-1}=I_r+O_p\left(\delta_{N T_{0}}^{-2}\right)$, then
			\[
			\sqrt{T_{0}}\left(\widehat{\boldsymbol{\lambda}}_i-\boldsymbol{\lambda}_i\right)
			=
			\left(\frac{\boldsymbol{F}^{\prime} \boldsymbol{F}}{T_{0}}\right)^{-1}
			\frac{1}{\sqrt{T_{0}}} \sum_{t=1}^{T_{0}} \boldsymbol{F}_t \,\xi_{i t}
			+o_p(1).
			\]
			
			In particular, they prove that when $N, T_{0} \rightarrow \infty$ and $\sqrt{T_{0}}/N \rightarrow 0$,
			\[
			\sqrt{T_{0}}\left(\widehat{\boldsymbol{\lambda}}_i-\boldsymbol{\lambda}_i\right) 
			\xrightarrow{d} \mathcal{N} \left(0, \boldsymbol{\Phi}_i\right),
			\]
			
			where
			\[
			\boldsymbol{\Phi}_i
			=
			\lim _{T_{0} \rightarrow \infty} \frac{1}{T_{0}} 
			\sum_{t, s=1}^{T_{0}} 
			\mathbb{E}\left[\boldsymbol{F}_t \boldsymbol{F}_s^{\prime} \xi_{i t} \xi_{i s}\right].
			\]
			
			Leveraging these results and multipling and dividing the second term of equation \eqref{equation:beta_exp} inside the parentheses by $\sqrt{T_0}$, we obtain  
			
			\begin{equation} 
				\label{eq:beta-expansion-final}
				\begin{aligned}
					\sqrt{N}(\widehat{\boldsymbol{\beta}}-\boldsymbol{\beta})
					= {}&
					\widehat{\boldsymbol{E}}_{\beta \beta}^{-1}\,
					\frac{1}{\sqrt{N}}\, \Bigg( \sum_{i=1}^{N}\boldsymbol{S}_i(\boldsymbol{\beta},\boldsymbol{\lambda}_i)
					\\
					&\quad
					+ \frac{1}{\sqrt{T_{0}}} \sum_{i=1}^{N} 
					\frac{\partial \boldsymbol{S}_{i}\left(\boldsymbol{\beta},\boldsymbol{\lambda}_i \right)}
					{\partial \boldsymbol{\lambda}^{\mathcal{T}}}
					\left(\frac{ \boldsymbol{F}^{\prime} \boldsymbol{F}}{T_{0}}\right)^{-1}
					\frac{1}{\sqrt{T_{0}}} \sum_{t=1}^{T_{0}} \boldsymbol{F}_t \,\xi_{i t}
					\Bigg).
				\end{aligned} 
			\end{equation}
			
			Then, we write equation \eqref{eq:beta-expansion-final} compactly as 
			
			\[
			\sqrt{N}(\widehat{\boldsymbol{\beta}}-\boldsymbol{\beta})
			=
			\widehat{\boldsymbol{E}}_{\beta \beta}^{-1}\,
			\frac{1}{\sqrt{N}}\sum_{i=1}^{N} \boldsymbol{\psi}_i.
			\]
			
			where 
			
			$$\boldsymbol{\psi}_i
			=
			\boldsymbol{S}_i(\boldsymbol{\beta},\boldsymbol{\lambda}_i)
			+
			\frac{1}{\sqrt{T_0}}\frac{\partial \boldsymbol{S}_{i}\left(\boldsymbol{\beta},\boldsymbol{\lambda}_i \right)}
			{\partial \boldsymbol{\lambda}^{\mathcal{T}}}
			\left(\frac{ \boldsymbol{F}^{\prime} \boldsymbol{F}}{T_{0}}\right)^{-1}
			\frac{1}{\sqrt{T_{0}}} \sum_{t=1}^{T_{0}} \boldsymbol{F}_t \,\xi_{i t}.$$
			
			Applying Slutsky's theorem combined with the results of \cite{bai2013principal}, we obtain
			
			\[
			\sqrt{N}(\widehat{\boldsymbol{\beta}}-\boldsymbol{\beta})
			=
			{\boldsymbol{E}}_{\beta \beta}^{-1}\,
			\frac{1}{\sqrt{N}}\sum_{i=1}^{N} \boldsymbol{\psi}_i + o_p(1),
			\]
			
			where $\boldsymbol{E}_{\beta \beta} = {-} \mathbb{E}[\frac{\partial \boldsymbol{S}_{i}\left(\boldsymbol{\beta},\boldsymbol{\lambda}_i \right)}{\partial \boldsymbol{\beta}^{\mathcal{T}}}]$. 
			
			Under standard regularity conditions, we have 
			
			\[
			\sqrt{N}\,(\widehat{\boldsymbol{\beta}} - \boldsymbol{\beta})
			\ \xrightarrow{d}\
			\mathcal{N}\!\left(\boldsymbol{0},\,\mathbb{V}_\beta\right),
			\qquad
			\mathbb{V}_\beta
			=
			\boldsymbol{E}_{\beta\beta}^{-1}\,
			\boldsymbol{D}_0\,
			\boldsymbol{E}_{\beta\beta}^{-1}.
			\]
			
			where $\boldsymbol{D}_0 = \mathbb{E}\!\left[\boldsymbol{\psi}_i\,\boldsymbol{\psi}_i^{\mathcal{T}}\right].$
			
			We additionally assume asymptotic orthogonality between the second-step score
			innovations and the first-step PCA estimation error, so that the corresponding
			cross term in $\boldsymbol{D}_0$ vanishes. Consequently,
			
			\begin{equation}
				\label{eq:D0-decomposition}
				\mathbb{E}\big[\boldsymbol{\psi}_i\,\boldsymbol{\psi}_i'\big]
				=
				\mathbb{E}\big[
				\boldsymbol{S}_i(\boldsymbol{\beta},\boldsymbol{\lambda}_i)\,
				\boldsymbol{S}_i(\boldsymbol{\beta},\boldsymbol{\lambda}_i)^{\mathcal{T}}
				\big]
				+
				\mathbb{E}\!\Big[
				\frac{\partial \boldsymbol{S}_{i}\left(\boldsymbol{\beta},\boldsymbol{\lambda}_i \right)}
				{\partial \boldsymbol{\lambda}^{\mathcal{T}}}\,
				\frac{\boldsymbol{\Phi}_i}{T_0}\,
				\frac{\partial \boldsymbol{S}_{i}\left(\boldsymbol{\beta},\boldsymbol{\lambda}_i \right)}
				{\partial \boldsymbol{\lambda}^{\mathcal{T}}}^{\mathcal{T}}
				\Big].
			\end{equation}
			
			By the asymptotic results of \citet{bai2013principal}, the second term in
			\eqref{eq:D0-decomposition} is of order $O(1/T_0)$ and therefore vanishes as
			$T_0 \to \infty$. Hence, in the limit,
			\[
			\boldsymbol{D}_0
			=
			\mathbb{E}\big[\boldsymbol{\psi}_i\,\boldsymbol{\psi}_i^{\mathcal{T}}\big]
			=
			\mathbb{E}\big[\boldsymbol{S}_i(\boldsymbol{\beta},\boldsymbol{\lambda}_i)
			\,\boldsymbol{S}_i(\boldsymbol{\beta},\boldsymbol{\lambda}_i)^{\mathcal{T}}\big].
			\]}
		This completes the proof.
		
	\end{proof}
	
	To prove Theorem~\ref{thm1}, recall that the population ATT parameter is
	
	\begin{equation}
		\tau^{ATT} = \tau_1 - \tau_0,
		\label{equation 11}
	\end{equation}
	
	where
	\[
	\tau_1 = \mathbb{E}\big[ Y_{i,T}(1)\mid Z_{i,T}=1 \big], 
	\qquad
	\tau_0 = \mathbb{E}\big[ Y_{i,T}(0)\mid Z_{i,T}=1 \big].
	\]
	
	The corresponding sample ATT estimator is
	\begin{equation} 
		\begin{aligned}
			\hat{\tau}^{ATT} 
			&= \hat{\tau}_1 - \hat{\tau}_0 \\
			&= \frac{\sum_{i=1}^{N} Z_{i,T}\,Y_{i,T}}{\sum_{i=1}^{N} Z_{i,T}} 
			\;-\; 
			\frac{\sum_{i=1}^{N} (1-Z_{i,T})\,Y_{i,T}\,\bigl(\hat{e}_{i}/(1-\hat{e}_{i})\bigr)}
			{\sum_{i=1}^{N} (1-Z_{i,T})\,\bigl(\hat{e}_{i}/(1-\hat{e}_{i})\bigr)}.
		\end{aligned} 
		\label{equation 10}
	\end{equation}
	
	From the first term in \eqref{equation 10}, we derive the estimating equation
	\begin{equation} 
		\sum_{i=1}^{N} U_{i}\left(\hat{\tau}_1\right)
		= \sum_{i=1}^{N} Z_{i,T}\left(Y_{i,T}-\hat{\tau}_1\right) = 0,
		\label{eq8}
	\end{equation}
	where $U_i(\tau_1) = Z_{i,T}(Y_{i,T}-\tau_1)$.
	The derivative with respect to $\tau_1$ is
	\[
	\frac{\partial U_{i}\left(\tau_1\right)}{\partial \tau_1} = -Z_{i,T}.
	\]
	
	A first–order mean-value expansion of \eqref{eq8}
	around the true value $\tau_1$ yields
	
	\begin{equation*} 
		\begin{aligned}
			0 
			&= \frac{1}{\sqrt{N}} \sum_{i=1}^{N} U_{i}\left(\hat{\tau}_1\right) \\
			&= \frac{1}{\sqrt{N}} \sum_{i=1}^{N} U_{i}\left(\tau_1\right)
			+ \Bigg(
			\frac{1}{N} \sum_{i=1}^{N} 
			\left.\frac{\partial U_{i}\left(\tau_1\right)}{\partial \tau_1}\right|_{\tau_1=\tilde{\tau}_1}
			\Bigg)\,
			\sqrt{N}\left(\hat{\tau}_1-\tau_1\right),
		\end{aligned} 
	\end{equation*}
	
	where $\tilde{\tau}_1$ lies between $\tau_1$ and $\hat{\tau}_1$.
	
	Solving for $\sqrt{N}(\hat{\tau}_1-\tau_1)$ gives
	\[
	\sqrt{N}\left(\hat{\tau}_1-\tau_1\right)
	= \hat\eta_{1}^{-1}\left(\frac{1}{\sqrt{N}} \sum_{i=1}^{N} U_{i}\left(\tau_1\right)\right)
	+ o_p(1), 
	\]
	
	where $\,\frac{1}{N}\sum_{i=1}^{N} Z_{i,T} = \,\hat\eta_1$. By Slutsky's theorem, we can equivalently write
	
	\begin{equation*} 
		\sqrt{N}\left(\hat{\tau}_1-\tau_1\right)
		= \frac{1}{\sqrt{N}} \sum_{i=1}^{N} I_{1i} + o_p(1),
		\qquad
		I_{1i} = \eta_{1}^{-1}\,Z_{i,T}\left(Y_{i,T}-\tau_1\right),
	\end{equation*} 
	
	where $\eta_{1} = \mathbb{E}\bigl[Z_{i,T}\bigr] = \Pr(Z_{i,T}= 1)$.
	
	From the second term in \eqref{equation 10}, we obtain the sample estimating
	equation for the counterfactual mean $\hat\tau_0$:
	\begin{equation} 
		\begin{aligned}
			\sum_{i=1}^{N} U_{i}\left(\hat{\tau}_0, \hat{\boldsymbol{\beta}}\right) 
			&= \frac{\sum_{i=1}^{N} (1-Z_{i,T})\,Y_{i,T}\,\bigl(\hat{e}_{i}/(1-\hat{e}_{i})\bigr)}
			{\sum_{i=1}^{N} (1-Z_{i,T})\,\bigl(\hat{e}_{i}/(1-\hat{e}_{i})\bigr)} 
			- \hat\tau_0  \\
			&= \sum_{i=1}^{N} (1-Z_{i,T})\,Y_{i,T}\,\bigl(\hat{e}_{i}/(1-\hat{e}_{i})\bigr)
			- \Bigl(\sum_{i=1}^{N} (1-Z_{i,T})\,\bigl(\hat{e}_{i}/(1-\hat{e}_{i})\bigr)\Bigr)\hat{\tau}_0 \\
			&= \sum_{i=1}^{N} (1-Z_{i,T})\,\bigl(\hat{e}_{i}/(1-\hat{e}_{i})\bigr)\,\bigl(Y_{i,T}-\hat{\tau}_0\bigr) = 0,
		\end{aligned} 
		\label{eq12}
	\end{equation}
	
	The gradient is given by 
	
	\begin{equation*} 
		\begin{aligned}
			\frac{\partial U_i\left(\tau_0, \boldsymbol{\beta}\right)}{\partial \tau_0}
			&= - (1-Z_{i,T})\,\frac{e_{i}}{1-e_{i}}, \\
			\frac{\partial U_i\left(\tau_0, \boldsymbol{\beta}\right)}{\partial \boldsymbol{\beta}}
			&= (1-Z_{i,T})\, \frac{e_{\boldsymbol{\beta},i}}{(1-e_i)^2}\,\bigl(Y_{i,T}-\tau_0\bigr),
		\end{aligned} 
	\end{equation*}
	
	where $e_{\boldsymbol{\beta},i} \equiv \partial e_{i}/\partial \boldsymbol{\beta}$.
	A first–order mean-value expansion of \eqref{eq12} around the true values
	$(\tau_0,\boldsymbol{\beta})$ gives
	\begin{equation*} 
		\begin{aligned}
			\frac{1}{\sqrt{N}} \sum_{i=1}^{N} U_{i}\left(\hat{\tau}_0, \hat{\boldsymbol{\beta}} \right) 
			=&\ \frac{1}{\sqrt{N}} \sum_{i=1}^{N} U_{i}\left(\tau_0, \boldsymbol{\beta}\right) \\
			&\ -\Bigg(\frac{1}{N} \sum_{i=1}^{N} (1-Z_{i,T})\,\frac{\tilde{e}_{i}}{1-\tilde{e}_{i}}\Bigg)\,
			\sqrt{N}\left(\hat{\tau}_0-\tau_0\right) \\
			&\ +\Bigg(\frac{1}{N} \sum_{i=1}^{N}  (1-Z_{i,T})\,
			\frac{\tilde{e}_{\boldsymbol{\beta},i}}{(1-\tilde{e}_i)^2}\,\bigl(Y_{i,T}-\tilde{\tau}_0\bigr)\Bigg)^{\!\top}
			\sqrt{N}\bigl(\hat{\boldsymbol{\beta}}-\boldsymbol{\beta}\bigr) =\ 0,
		\end{aligned} 
	\end{equation*}
	
	where $(\tilde{\tau}_0 , \tilde{\boldsymbol{\beta}})$ is a mean value between
	$(\tau_0 , \boldsymbol{\beta})$ and $(\hat{\tau}_0 , \hat{\boldsymbol{\beta}})$,
	$\tilde{e}_{i} \equiv e(\boldsymbol{\hat{\lambda}}_i , \tilde{\boldsymbol{\beta}})$, and
	$\tilde{e}_{\boldsymbol{\beta},i} \equiv
	\left.\partial e_{i}/\partial \boldsymbol{\beta}\right|_{\boldsymbol{\beta}=\tilde{\boldsymbol{\beta}}}$.
	
	Define the population quantities
	\[
	\eta_{2} = \mathbb{E}\Bigl[(1-Z_{i,T})\,\frac{e_{i}}{1-e_{i}}\Bigr],
	\qquad
	\boldsymbol{H}_\beta 
	= \mathbb{E} \Bigl[(1-Z_{i,T})\,\frac{e_{\boldsymbol{\beta},i}}{(1-e_i)^2}\,
	\bigl(Y_{i,T}-\tau_0\bigr) \Bigr],
	\]
	and let $\hat\eta_2$ and $\hat{\boldsymbol{H}}_\beta$ denote the corresponding
	sample averages. Solving the previous expansion for $\sqrt{N}(\hat\tau_0-\tau_0)$, we obtain
	\[
	\sqrt{N}\left(\hat{\tau}_0-\tau_0\right) 
	= \hat{\eta}_{2}^{-1}\left(\frac{1}{\sqrt{N}} \sum_{i=1}^{N} U_{i}\left(\tau_0, \boldsymbol{\beta}\right)\right)
	+ \hat{\eta}_{2}^{-1}\,\hat{\boldsymbol{H}}_\beta^{\top}\sqrt{N}\bigl(\hat{\boldsymbol{\beta}}-\boldsymbol{\beta}\bigr)
	+ o_p(1).
	\]
	
	By Lemma~\ref{lemma1}, the propensity-score estimator satisfies the linear
	representation
	\[
	\sqrt{N}(\widehat{\boldsymbol{\beta}}-\boldsymbol{\beta})
	=
	\widehat{\boldsymbol{E}}_{\beta \beta}^{-1}\,
	\frac{1}{\sqrt{N}}\sum_{i=1}^{N} \boldsymbol{\psi}_i,
	\]
	for some influence function $\boldsymbol{\psi}_i$ defined in
	Lemma~\ref{lemma1}. Using Slutsky's theorem, we can therefore write
	\[
	\sqrt{N}\left(\hat{\tau}_0-\tau_0\right) 
	= \frac{1}{\sqrt{N}} \sum_{i=1}^{N} I_{0i} + o_p(1),
	\qquad
	I_{0i} = \eta_{2}^{-1}\Bigl[
	U_i(\tau_0,\boldsymbol{\beta})
	+ \boldsymbol H_\beta^\top
	{\boldsymbol E_{\beta\beta}^{-1}}
	\boldsymbol\psi_i
	\Bigr].
	\]
	
	Combining the influence–function representations for $\hat\tau_1$ and
	$\hat\tau_0$, we obtain
	
	\begin{equation*}
		\begin{aligned}
			\sqrt{N}\left(\hat{\tau}^{ATT}-\tau^{ATT}\right) 
			&=\sqrt{N}\left(\hat{\tau}_1-\tau_1\right)
			-\sqrt{N}\left(\hat{\tau}_0-\tau_0\right) \\
			&= \frac{1}{\sqrt{N}} \sum_{i=1}^{N} I_{i}^{\text{full}} + o_p(1),
		\end{aligned}
	\end{equation*}
	
	where
	\[
	I_i^{\text{full}} = I_{1i} - I_{0i}
	= \eta_{1}^{-1}\,U_{i}\left(\tau_1\right)
	- \eta_{2}^{-1}\,U_{i}\left(\tau_0, \boldsymbol{\beta}\right)
	- \eta_{2}^{-1} \boldsymbol{H}_\beta^{\top} \,{\boldsymbol{E}}_{\beta \beta}^{-1}\,
	\boldsymbol{\psi}_i.
	\]
	
	Under the regularity conditions in Assumptions~1–4 and standard regularity condition of M-estimation (\citealp[]{wooldridge2010econometric}, Chapter 12), a joint central limit
	theorem applies to $(U_i(\tau_1),U_i(\tau_0,\boldsymbol{\beta}),\boldsymbol{\psi}_i)$,
	so that
	\[
	\sqrt{N}\left(\hat{\tau}^{ATT}-\tau^{ATT}\right)
	\overset{d}{\longrightarrow}
	\mathcal{N}\bigl(0,\,\mathbb{E}\bigl[I_i^{\text{full}} I_i^{\text{full}\prime}\bigr]\bigr).
	\]
	
	By Lemma~\ref{lemma1}, the influence function of the propensity-score
	estimator can be written as
	\[
	\boldsymbol{\psi}_i
	=
	\boldsymbol{S}_i(\boldsymbol{\beta},\boldsymbol{\lambda}_i)
	+ \boldsymbol{r}_i,
	\]
	where
	\[
	\boldsymbol{r}_i
	=
	\frac{1}{\sqrt{T_0}}\,
	\frac{\partial \boldsymbol{S}_{i}\left(\boldsymbol{\beta},\boldsymbol{\lambda}_i \right)}
	{\partial \boldsymbol{\lambda}^{\mathcal{T}}}
	\left(\frac{ \boldsymbol{F}^{\prime} \boldsymbol{F}}{T_{0}}\right)^{-1}
	\frac{1}{\sqrt{T_{0}}} \sum_{t=1}^{T_{0}} \boldsymbol{F}_t \,\xi_{it}.
	\]
	
	Substituting this decomposition into the expression of $I_i^{\text{full}}$ gives
	\[
	I_i^{\text{full}}
	=
	\eta_{1}^{-1}\,U_i(\tau_1)
	-
	\eta_{2}^{-1}\,U_i(\tau_0,\boldsymbol{\beta})
	-
	\eta_{2}^{-1}\,\boldsymbol{H}_\beta^{\top}\,\boldsymbol{E}_{\beta \beta}^{-1}
	\,\boldsymbol{S}_i(\boldsymbol{\beta},\boldsymbol{\lambda}_i)
	-
	\eta_{2}^{-1}\,\boldsymbol{H}_\beta^{\top}\,\boldsymbol{E}_{\beta \beta}^{-1}
	\,\boldsymbol{r}_i.
	\]
	Define
	\begin{equation}
		I_i
		=
		\eta_{1}^{-1}\,U_i(\tau_1)
		-
		\eta_{2}^{-1}\,U_i(\tau_0,\boldsymbol{\beta})
		-
		\eta_{2}^{-1}\,\boldsymbol{H}_\beta^{\top}\,\boldsymbol{E}_{\beta \beta}^{-1}
		\,\boldsymbol{S}_i(\boldsymbol{\beta},\boldsymbol{\lambda}_i)
		\label{equation 20}
	\end{equation}
	
	and
	
	\[
	R_i
	=
	- \eta_{2}^{-1}\,\boldsymbol{H}_\beta^{\top}\,\boldsymbol{E}_{\beta \beta}^{-1}
	\,\boldsymbol{r}_i,
	\]
	so that $I_i^{\text{full}} = I_i + R_i$. 
	Under the asymptotic results for approximate factor models in
	\citet{bai2013principal}, the PCA error in the factor loadings satisfies
	$\mathbb{E}[\boldsymbol{r}_i \boldsymbol{r}_i^{\top}] = O(1/T_0)$ as
	$T_0 \to \infty$. Since $R_i$ is a fixed linear combination of the
	components of $\boldsymbol{r}_i$, this implies
	$\mathbb{E}[R_i R_i'] = O(1/T_0)$ and the cross-product term
	$\mathbb{E}[I_i R_i'] = O(1/\sqrt{T_0})$. Hence
	\[
	\mathbb{E}\bigl[I_i^{\text{full}} I_i^{\text{full}\prime}\bigr]
	=
	\mathbb{E}[I_i I_i'] + o(1)
	\quad \text{as } T_0 \to \infty.
	\]
	Consequently, the asymptotic variance of
	$\sqrt{N}(\hat{\tau}^{ATT}-\tau^{ATT})$ is $\mathbb{E}[I_i I_i']$. This completes the proof.
	
	\subsection{Feasible estimation of the asymptotic variance} \label{app: empirical}
	
	The empirical M-estimation variance of $\hat\tau^{ATT}$ is obtained by considering the empirical counterpart of the influence-function representation in Equation \eqref{equation 20}. Let $\hat I_i$ denote $I_i$ evaluated at the sample analogues of all population quantities. Under the conditions of Theorem~\ref{thm1} and Lemma~\ref{lemma1}, a feasible variance estimator is
	\[
	\widehat{\mathrm{Var}}(\hat\tau^{ATT})
	= N^{-2}\sum_{i=1}^N \hat I_i^{\,2}.
	\]
	
	with 
	
	\[
	\hat I_i
	=
	\hat\eta_1^{-1}\hat U_i(\hat\tau_1)
	-
	\hat\eta_2^{-1}\hat U_i(\hat\tau_0,\hat\beta)
	-
	\hat\eta_2^{-1}\hat{ \boldsymbol{H}}_\beta^{\top}\hat{ \boldsymbol{E}}_{\beta\beta}^{-1}\hat{ \boldsymbol{S}}_i.
	\]
	
	The quantities $\tau_1$ and $\tau_0$ in Equation \eqref{equation 11} are consistently estimated via the H\'ajek IPW estimators
	\[
	\hat\tau_1
	= \frac{\sum_{i=1}^N Z_{i,T}Y_{i,T}}{\sum_{i=1}^N Z_{i,T}},
	\qquad
	\hat\tau_0
	= \frac{\sum_{i=1}^N (1-Z_{i,T})Y_{i,T}\,\hat e_i/(1-\hat e_i)}
	{\sum_{i=1}^N (1-Z_{i,T})\,\hat e_i/(1-\hat e_i)}.
	\]
	where $\hat e_i$ is the fitted propensity score from the logistic regression of $Z_{i,T}$ on the estimated factor loadings $\hat\lambda_i$. Accordingly, $U_i(\tau_1)$ and $U_i(\tau_0,\beta)$ in $I_i$ are replaced by
	\[
	\hat U_i(\hat\tau_1)= Z_{i,T}(Y_{i,T}-\hat\tau_1),
	\qquad
	\hat U_i(\hat\tau_0,\hat\beta)
	= (1-Z_{i,T})\frac{\hat e_i}{1-\hat e_i}\,(Y_{i,T}-\hat\tau_0).
	\]
	
	Moreover, the normalizing constants $\eta_1=\mathbb{E}[Z_{i,T}]$ and $\eta_2=\mathbb{E}\!\left[(1-Z_{i,T})\,\frac{e_i}{1-e_i}\right]$ are consistently estimated by the corresponding sample averages
	
	\[
	\hat\eta_1=\frac{1}{N}\sum_{i=1}^N Z_{i,T},
	\qquad
	\hat\eta_2=\frac{1}{N}\sum_{i=1}^N (1-Z_{i,T})\,\frac{\hat e_i}{1-\hat e_i}.
	\]
	
	Finally, the empirical score of the propensity-score model with respect to $\beta$ is
	$\hat {\boldsymbol{S}}_i = \boldsymbol{S}_i(\hat\beta,\hat\lambda_i)$. The matrices $\boldsymbol{E}_{\beta\beta}$ and $\boldsymbol{H}_\beta$ are
	consistently estimated by their sample counterparts, $\hat{\boldsymbol{E}}_{\beta\beta}$ and $\hat{ \boldsymbol{H}}_\beta$, computed as the
	sample means of the terms in in Lemma \ref{lemma1} and Theorem~\ref{thm1} evaluated at $(\hat\beta,\hat\lambda_i)$.

	
	
	\subsection{Preliminar check for the falsification tests}\label{app: sec13}
	
	For the three falsification tests, Figures \ref{fig:placebo_ps_loadings} and \ref{fig:Figure_PlA_Prop_Score_Z_L_X} show the estimated propensity score distributions by treatment group under the baseline specification and the augmented specification including pre-treatment covariates. Tables \ref{tab:asd_table_fict_loadings} and \ref{tab:asd_with_covariates_fict} report the corresponding ASD values for the original and weighted samples. Overall, we observe a satisfactory degree of overlap in the estimated propensity scores and ASD values below the critical threshold of 1.96 after weighting.
	
	\begin{figure}[!htbp]
		\centering
		\includegraphics[width=0.47\textwidth]{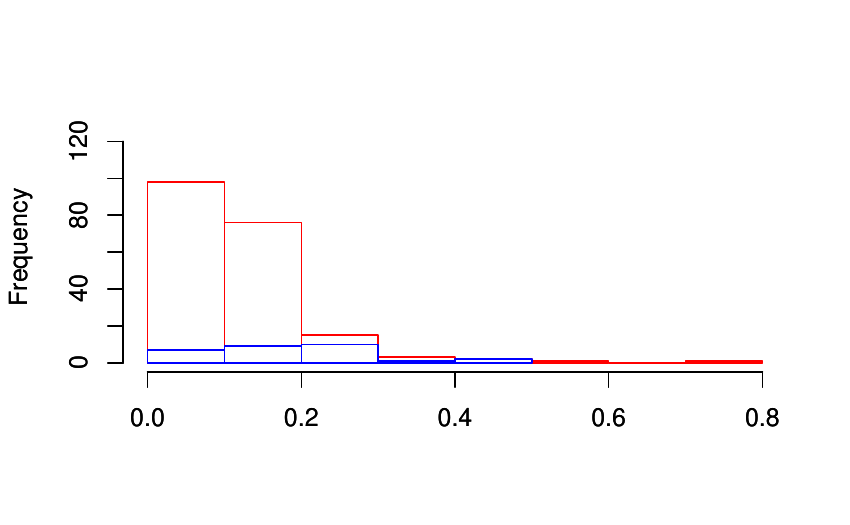}\hfill
		\includegraphics[width=0.47\textwidth]{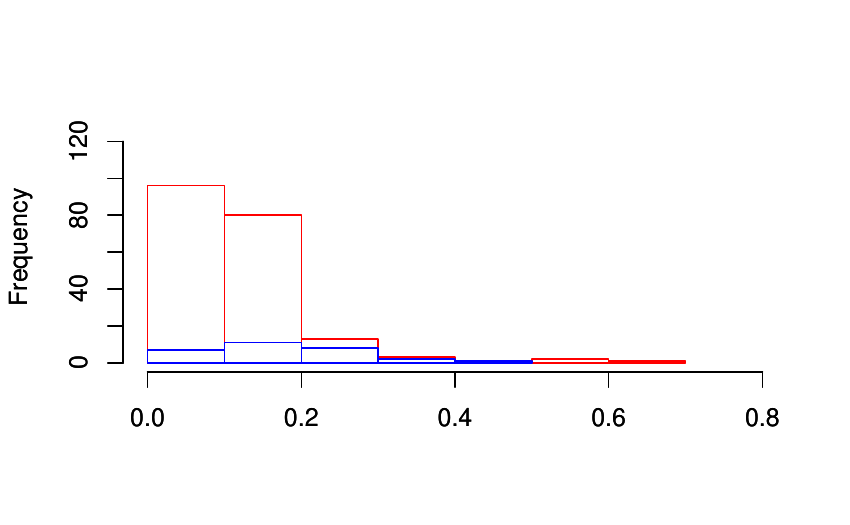}
		
		\vspace{0.8em}
		
		\includegraphics[width=0.47\textwidth]{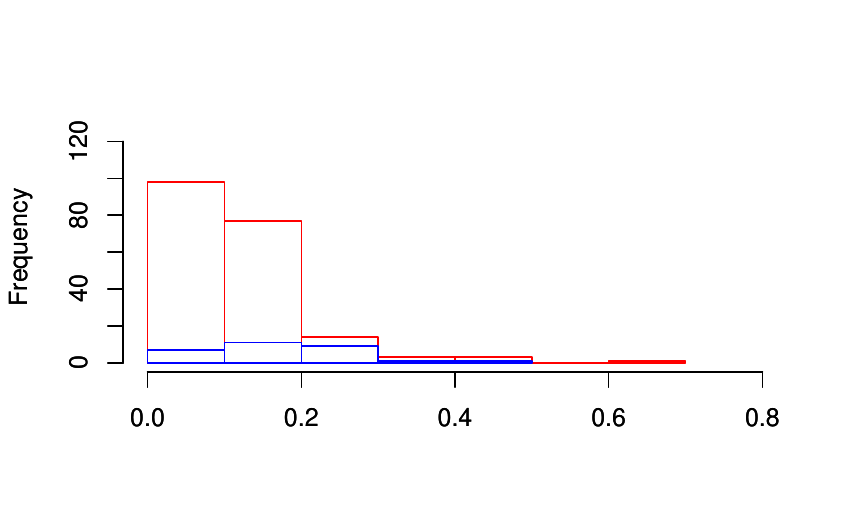}
		
		\caption{Estimated propensity score distributions for treated and untreated firms in the three falsification tests under the specification with factor loadings only. The panels refer, respectively, to fictitious implementation dates 01/04/2011 (top left), 01/04/2013 (top right), and 01/04/2015 (bottom).}
		\label{fig:placebo_ps_loadings}
	\end{figure}
	
	\begin{figure}[!htbp]
		\centering
		\includegraphics[width=0.47\textwidth]{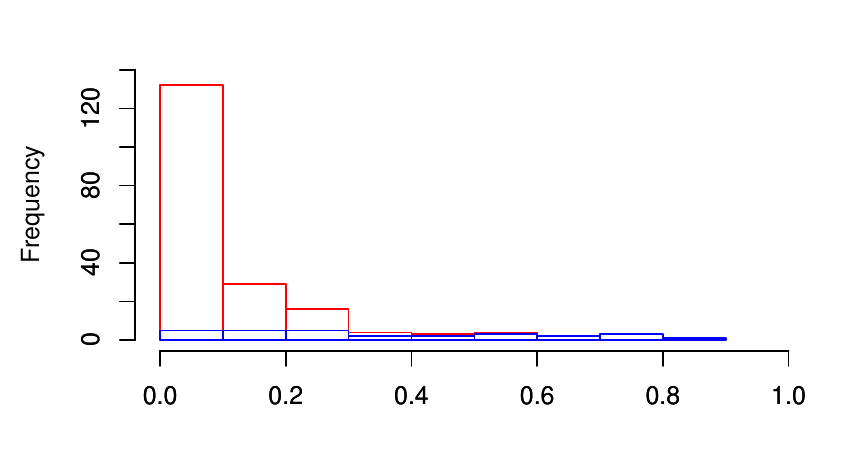}\hfill
		\includegraphics[width=0.47\textwidth]{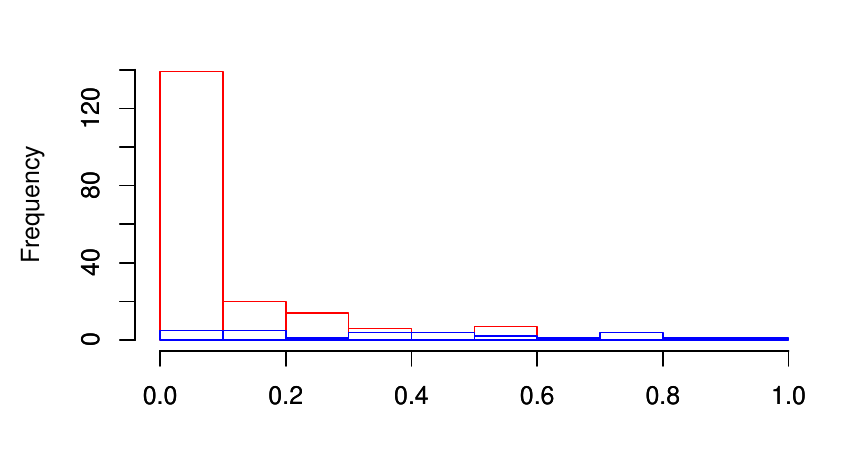}
		
		\vspace{0.8em}
		
		\includegraphics[width=0.47\textwidth]{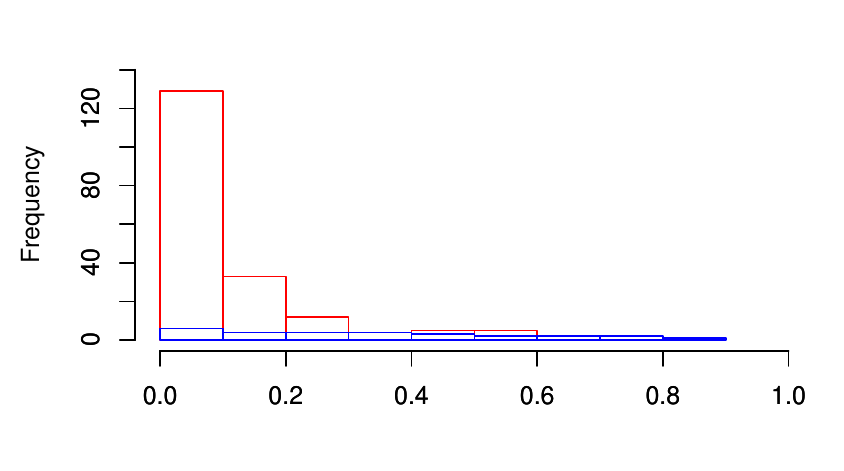}
		
		\caption{Histograms of the estimated propensity score for the treated (blue) and the untreated (red) under the augmented specification with factor loadings and pre-treatment covariates. The panels refer, respectively, to fictitious implementation dates 01/04/2011 (top left), 01/04/2013 (top right), and 01/04/2015 (bottom).}
		\label{fig:Figure_PlA_Prop_Score_Z_L_X}
	\end{figure}
	
	\begin{table}[!th]
		\centering
		\small{
			\begin{tabular}{l|cc|cc|cc}
				\hline
				\multicolumn{1}{l|}{Fictitious date} & \multicolumn{2}{c|}{01/04/2011} & \multicolumn{2}{c|}{01/04/2013} & \multicolumn{2}{c}{01/04/2015} \\
				& Unweighted & Weighted & Unweighted & Weighted & Unweighted & Weighted \\
				\hline
				\multicolumn{1}{c|}{${\hat{\lambda}}_1$} & 1.63 & 0.49 & 1.66 & 0.37 & 1.70 & 0.33 \\
				\multicolumn{1}{c|}{${\hat{\lambda}}_2$} & 2.58 & 1.11 & 2.43 & 0.64 & 2.45 & 0.61 \\
				\multicolumn{1}{c|}{${\hat{\lambda}}_3$} & 2.77 & 0.60 & 2.86 & 0.46 & 2.75 & 0.41 \\
				\hline
			\end{tabular}
			\medskip
			\caption{Absolute standardized differences of the estimated factor loadings in the original and weighted samples for the three falsification tests under the specification with factor loadings only.}
			\label{tab:asd_table_fict_loadings}
		}
	\end{table}
	
	\begin{table}[!htbp]
		\centering
		\small
		\setlength{\tabcolsep}{4pt}
		\renewcommand{\arraystretch}{1.05}
		\begin{tabular}{l|cc|cc|cc}
			\hline
			\multicolumn{1}{c|}{Fictitious date} & \multicolumn{2}{c|}{01/04/2011} & \multicolumn{2}{c|}{01/04/2013} & \multicolumn{2}{c}{01/04/2015} \\
			& Unweighted & Weighted & Unweighted & Weighted & Unweighted & Weighted \\
			\hline
			\multicolumn{1}{c|}{${\hat{\lambda}}_1$} & 1.440 & 0.021 & 1.494 & 0.107 & 1.495 & 0.030 \\
			\multicolumn{1}{c|}{${\hat{\lambda}}_2$} & 2.638 & 0.359 & 2.476 & 0.110 & 2.495 & 0.153 \\
			\multicolumn{1}{c|}{${\hat{\lambda}}_3$} & 2.534 & 0.134 & 2.554 & 0.051 & 2.461 & 0.185 \\
			\multicolumn{1}{c|}{Log(Revenue)} & 6.686 & 0.408 & 6.860 & 0.067 & 6.389 & 0.157 \\
			\multicolumn{1}{c|}{ROE} & 1.650 & 0.209 & 2.759 & 0.359 & 1.656 & 0.066 \\
			\multicolumn{1}{c|}{$\frac{\text{Total Debt}}{\text{Total Assets}}$} & 4.401 & 0.571 & 4.288 & 0.280 & 3.655 & 0.138 \\
			\multicolumn{1}{c|}{ESG score} & 2.088 & 0.836 & 2.103 & 1.436 & 2.088 & 1.051 \\
			\hline
		\end{tabular}
		\caption{Absolute standardized differences of the estimated factor loadings and selected pre-treatment covariates in the original and weighted samples for the three falsification tests under the augmented specification with factor loadings and covariates.}
		\label{tab:asd_with_covariates_fict}
	\end{table}
	
	\FloatBarrier
	
	\subsection{Preliminar check for the robustness analysis}\label{app: sec13_bis}
	
	This appendix reports additional descriptive information on the pre-treatment sample and helps document observable differences between treated and control firms. Table \ref{tab:characteristics} presents pre-treatment covariate means and balance diagnostics.
	
	\begin{table}[htbp]
		\centering
		\renewcommand{\arraystretch}{1.25}
		\resizebox{\textwidth}{!}{%
			\begin{tabular}{l l c c c c c c}
				\hline
				Variable & Category & Mean Treated & Mean Control & $ASD^{\text{unw.}}$ & $ASD^{L}$ & $ASD^{L+X}$ & N \\
				\hline
				\multicolumn{1}{c}{Interest Coverage} & Coverage & 2.08 & 71.44 & 2.91 & 2.09 & 1.10 & 175 \\
				\multicolumn{1}{c}{Log(Cash Flow from Operations)} & Coverage & 6.84 & 6.85 & 0.02 & 0.67 & 0.27 & 218 \\
				\multicolumn{1}{c}{ESG Score} & ESG & 2.75 & 3.16 & 2.12 & 1.98 & 0.94 & 218 \\
				\multicolumn{1}{c}{$\frac{\text{Cash Flow from Operations}}{\text{Total Debt}}$} & Leverage & 0.12 & 10.12 & 1.47 & 1.16 & 0.10 & 213 \\
				\multicolumn{1}{c}{$\frac{\text{Total Debt}}{\text{Total Assets}}$} & Leverage & 0.30 & 0.21 & 3.30 & 3.32 & 0.06 & 218 \\
				\multicolumn{1}{c}{$\frac{\text{Long-Term Debt}}{\text{Total Assets}}$} & Leverage & 0.21 & 0.16 & 1.83 & 2.15 & 0.33 & 218 \\
				\multicolumn{1}{c}{$\frac{\text{Net Debt}}{\text{EBITDA}}$} & Leverage & 1.59 & 1.44 & 0.24 & 0.05 & 0.96 & 182 \\
				\multicolumn{1}{c}{Current Ratio} & Liquidity & 1.13 & 1.49 & 2.33 & 2.17 & 0.38 & 175 \\
				\multicolumn{1}{c}{Quick Ratio} & Liquidity & 0.67 & 0.87 & 2.29 & 1.74 & 0.04 & 172 \\
				\multicolumn{1}{c}{EBITDA} & Profitability & 8007.00 & 4723.31 & 1.89 & 2.19 & 0.88 & 182 \\
				\multicolumn{1}{c}{EBITDA Margin} & Profitability & 0.20 & 0.19 & 0.13 & 0.40 & 0.06 & 45 \\
				\multicolumn{1}{c}{Adjusted EBITDA Margin} & Profitability & 0.17 & 0.21 & 0.32 & 0.33 & 0.50 & 33 \\
				\multicolumn{1}{c}{ROA} & Profitability & 0.01 & 0.06 & 7.57 & 5.39 & 1.93 & 218 \\
				\multicolumn{1}{c}{ROE} & Profitability & 0.07 & 0.15 & 3.10 & 1.71 & 0.20 & 218 \\
				\multicolumn{1}{c}{Log(Revenue)} & Size & 10.59 & 9.47 & 6.24 & 5.99 & 0.07 & 218 \\
				\hline
			\end{tabular}%
		}
		\caption{Means of pre-treatment covariates and absolute standardized differences before, $ASD^{unw.}$, and after propensity score weighting under alternative propensity score specifications based on loadings, $ASD^{L}$, and loadings augmented with covariates, $ASD^{L+X}$.}
		\label{tab:characteristics}
	\end{table}
	
	The following two tables provide additional descriptive information on the composition of the sample by country and sector. They show that treated and control firms differ not only in firm-level characteristics, but also in their observable composition across countries and sectors.
	
	\begin{table}[!th]
		\centering
		\small{
			\begin{tabular}{lrr}
				\hline
				Country & Treated (\%) & Control (\%)\\
				\hline
				Austria (AT) & 0.0 & 1.0\\
				Belgium (BE) & 3.4 & 2.0\\
				Bermuda (BM) & 0.0 & 0.5\\
				Denmark (DK) & 3.4 & 3.6\\
				Finland (FI) & 6.9 & 2.6\\
				France (FR) & 3.4 & 15.3\\
				Germany (DE) & 17.2 & 12.2\\
				Ireland (IE) & 0.0 & 3.1\\
				Italy (IT) & 6.9 & 3.1\\
				Jersey (JE) & 0.0 & 0.5\\
				Liechtenstein (LI) & 0.0 & 0.5\\
				Luxembourg (LU) & 0.0 & 1.0\\
				Netherlands (NL) & 6.9 & 6.1\\
				Norway (NO) & 0.0 & 1.5\\
				Portugal (PT) & 3.4 & 0.5\\
				Spain (ES) & 17.2 & 0.5\\
				Sweden (SE) & 13.8 & 8.7\\
				Switzerland (CH) & 6.9 & 10.7\\
				United Kingdom (GB) & 10.3 & 26.0\\
				United States (US) & 0.0 & 0.5\\
				\hline
			\end{tabular}
			\medskip
			\caption{Country composition of treated and control firms (pre-treatment).}
			\label{tab:country_comp}
			\medskip
			
			{\footnotesize \textit{Notes:} Country refers to the country of incorporation of the ultimate parent company. As a result, a small number of firms are incorporated outside the EU (e.g., Bermuda, Jersey, United States). Percentages are computed within treated and control groups and may not sum to 100 due to rounding.}
		}
	\end{table}
	
	\begin{table}[!th]
		\centering
		\small{
			\begin{tabular}{l r r}
				\hline
				GICS sector & Treated (\%) & Control (\%)\\
				\hline
				Communication Services & 6.9 & 5.1\\
				Consumer Discretionary & 10.3 & 11.7\\
				Consumer Staples & 0.0 & 13.3\\
				Energy & 0.0 & 2.6\\
				Financials & 55.2 & 14.3\\
				Health Care & 3.4 & 10.2\\
				Industrials & 6.9 & 22.4\\
				Information Technology & 0.0 & 6.6\\
				Materials & 3.4 & 9.2\\
				Real Estate & 0.0 & 1.5\\
				Utilities & 13.8 & 3.1\\
				\hline
			\end{tabular}
		}
		\medskip
		\caption{Sector composition of treated and control firms (pre-treatment).}
		\label{tab:sector_comp}
	\end{table}
	
	\FloatBarrier
	
	\subsection{Additional pre-treatment analysis for IPW}\label{app: sec13_ter}
	
	Figure \ref{fig:eventstudy} reports event-study plots of the estimated treatment effects in the pre-treatment period with estimated ATT point at each event time and confidence intervals, for both the specificiation with only factor loadings and the one augmented with the four selected pre-treatment covariates. Since the policy is implemented only at time $T$, all estimated ATT for $t<T$ should be statistically indistinguishable from zero under latent ignorability. To facilitate visual inspection, we also report inverse-variance weighted averages of the pre-treatment effects over time windows and overlay a nonparametric LOESS (Locally Estimated weighted Scatterplot Smoothing) curve for visual guidance. The plots show that the large majority of the estimated placebo ATTs are statistically indistinguishable from zero, supporting the validity of the identification strategy.
	
	\begin{figure}[!htbp]
		\centering
		\medskip
		\includegraphics[width=0.85\textwidth]{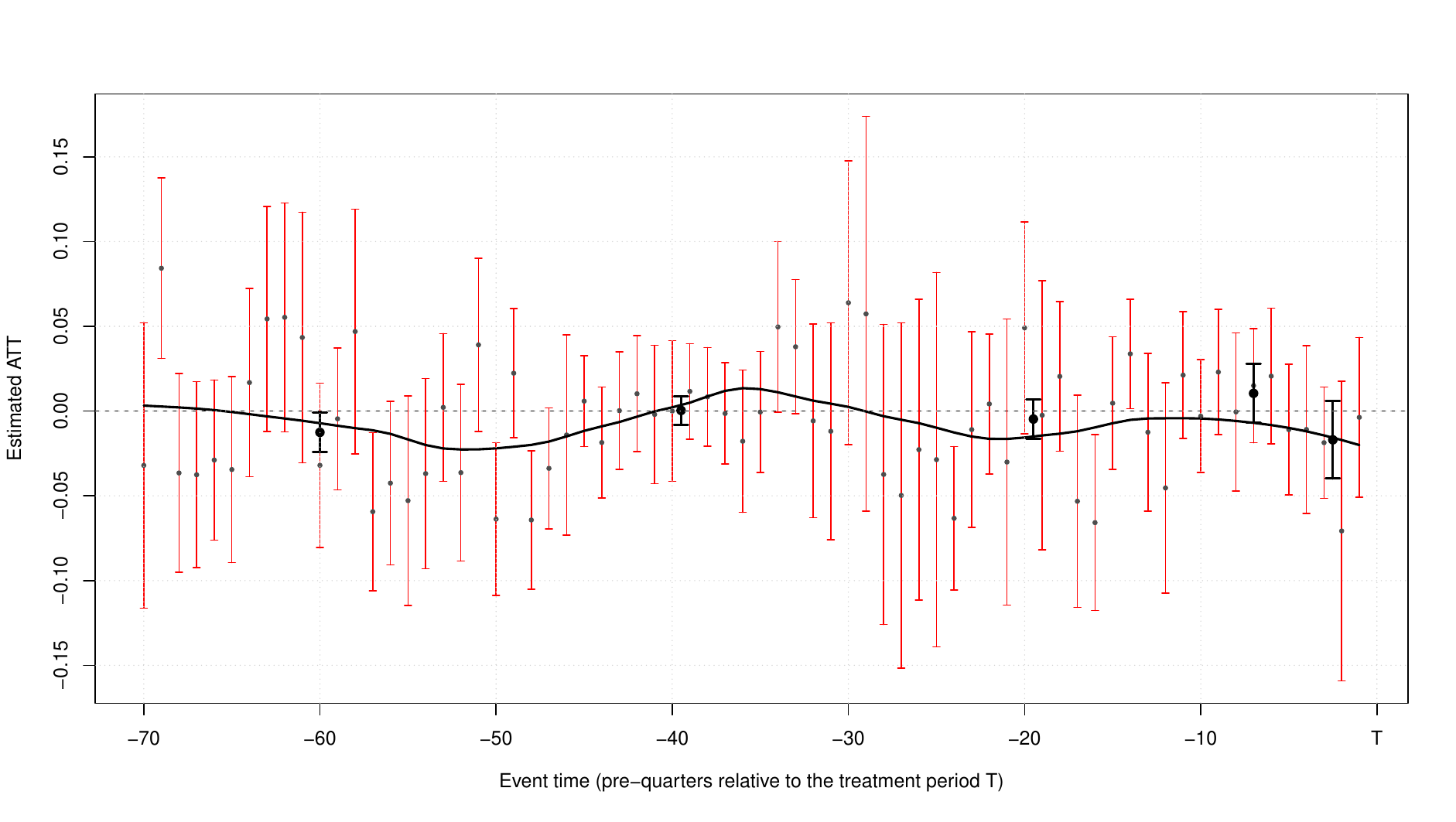}
		\vspace{0.6cm}
		\medskip
		\includegraphics[width=0.85\textwidth]{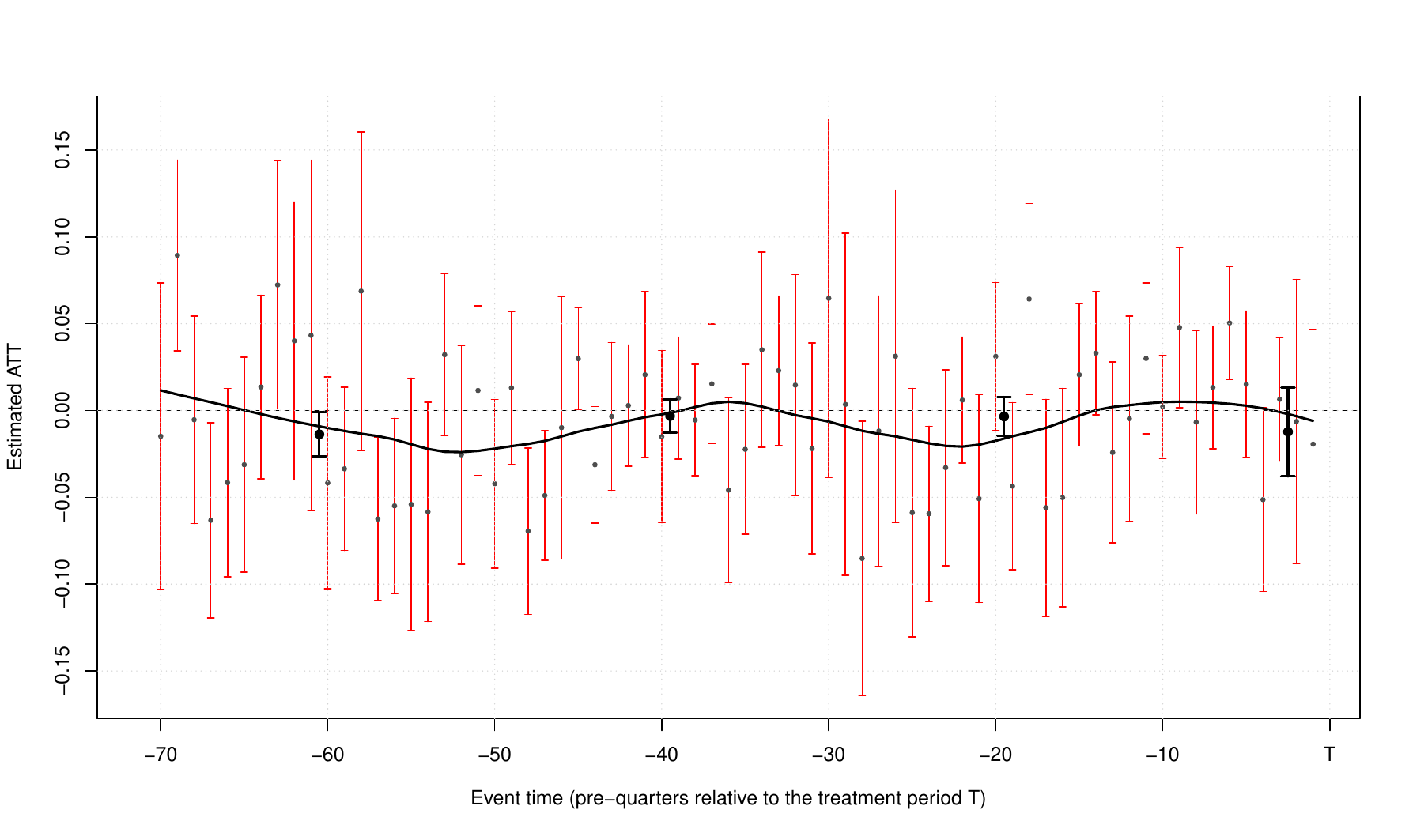}
		\caption{Event-study plots of placebo treatment effects in the pre-treatment period. The horizontal axis reports pre-quarters relative to the policy date $T$. Small points denote the estimated ATT at each pre-treatment quarter, with red vertical bars representing their 95\% confidence intervals. Larger points report inverse-variance weighted averages of the estimated ATT over time windows, with black confidence intervals. The solid black curve is a nonparametric LOESS smoother included for visual guidance only. The top plot reports results for the specification with only factor loadings, while the bottom one reports results for the specification augmented with pre-treatment covariates.}
		\label{fig:eventstudy}
	\end{figure}
	
	\end{appendices}
	
	\FloatBarrier
	
	\section{Competing interests}
	The authors confirm have no conflicts of interest to
	disclose.
	\section{Acknowledgments}
	
	
	
	We are grateful to Fabrizia Mealli, as well as to the participants of the Rome Workshop on Econometrics (September 2024, EIEF) and of the CIML Training and Conference ``Frontiers of Causal Inference and Machine Learning'' (July 2025, IMT, Lucca), for their helpful comments and suggestions. Andrea Mercatanti and Giacomo Morelli are funded by Sapienza grant "Causal-effect estimation of systemic risk" (grant No. RM123188F76ABDED).  Giacomo Morelli also acknowledges the financial support of the PRIN2022 grant (grant No. 2022ELYHCW).
	

\end{document}

\end{document}